\documentclass[11pt]{article}
\usepackage{verbatim,amsmath,eufrak,amsthm}
\usepackage{hyperref}
\usepackage{cite} % prevents showidx from correct work
\usepackage{yhmath}
\usepackage[toc,page]{appendix}

\hypersetup{
    bookmarks=true,         % show bookmarks bar?
    unicode=false,          % non-Latin characters in Acrobat's bookmarks
    pdftoolbar=true,        % show Acrobat's toolbar?
    pdfmenubar=true,        % show Acrobat's menu?
    pdffitwindow=false,     % window fit to page when opened
    pdfstartview={FitH},    % fits the width of the page to the window
    pdftitle={My title},    % title
    pdfauthor={Author},     % author
    pdfsubject={Subject},   % subject of the document
    pdfcreator={Creator},   % creator of the document
    pdfproducer={Producer}, % producer of the document
    pdfkeywords={keywords}, % list of keywords
    pdfnewwindow=true,      % links in new window
    colorlinks=false,       % false: boxed links; true: colored links
    linkcolor=red,          % color of internal links
    citecolor=green,        % color of links to bibliography
    filecolor=magenta,      % color of file links
    urlcolor=cyan           % color of external links
}
\renewcommand{\arraystretch}{1.2}
\makeatletter
\newdimen\normalarrayskip              % skip between lines
\newdimen\minarrayskip                 % minimal skip between lines
\normalarrayskip\baselineskip
\minarrayskip\jot
\newif\ifold             \oldtrue            \def\new{\oldfalse}
\def\arraymode{\ifold\relax\else\displaystyle\fi} % mode of array entries
\def\eqnumphantom{\phantom{(\theequation)}}     % right phantom in eqnarray
\def\@arrayskip{\ifold\baselineskip\z@\lineskip\z@
     \else
     \baselineskip\minarrayskip\lineskip2\minarrayskip\fi}
\def\@arrayclassz{\ifcase \@lastchclass \@acolampacol \or
\@ampacol \or \or \or \@addamp \or
   \@acolampacol \or \@firstampfalse \@acol \fi
\edef\@preamble{\@preamble
  \ifcase \@chnum
     \hfil$\relax\arraymode\@sharp$\hfil
     \or $\relax\arraymode\@sharp$\hfil
     \or \hfil$\relax\arraymode\@sharp$\fi}}
\def\@array[#1]#2{\setbox\@arstrutbox=\hbox{\vrule
     height\arraystretch \ht\strutbox
     depth\arraystretch \dp\strutbox
     width\z@}\@mkpream{#2}\edef\@preamble{\halign
\noexpand\@halignto
\bgroup \tabskip\z@ \@arstrut \@preamble \tabskip\z@ \cr}%
\let\@startpbox\@@startpbox \let\@endpbox\@@endpbox
  \if #1t\vtop \else \if#1b\vbox \else \vcenter \fi\fi
  \bgroup \let\par\relax
  \let\@sharp##\let\protect\relax
  \@arrayskip\@preamble}
\def\eqnarray{\stepcounter{equation}%
              \let\@currentlabel=\theequation
              \global\@eqnswtrue
              \global\@eqcnt\z@
              \tabskip\@centering
              \let\\=\@eqncr
 \halign to \displaywidth\bgroup
    \eqnumphantom\@eqnsel\hskip\@centering
    $\displaystyle \tabskip\z@ {##}$%
    \global\@eqcnt\@ne \hskip 2\arraycolsep
         $\displaystyle\arraymode{##}$\hfil
    \global\@eqcnt\tw@ \hskip 2\arraycolsep
         $\displaystyle\tabskip\z@{##}$\hfil
         \tabskip\@centering
    &{##}\tabskip\z@\cr}
\begingroup\ifx\undefined\newsymbol \else\def\input#1 {\endgroup}\fi

\newcounter{app}

\def\app{\setcounter{equation}{0}
\def\theequation{A\Roman{app}.\arabic{equation}}\par
   \addvspace{4ex}
   \@afterindentfalse
  \secdef\@app\@dapp}
\newcommand\@app{\@startsection {app}{1}{0ex}%
                                   {-3.5ex \@plus -1ex \@minus -.2ex}%
                                   {2.3ex \@plus.2ex}%
                                   {\normalfont\Large\bf}}

\def\@dapp#1{%
{\parindent \z@ \raggedright  \bf #1}\par\nobreak}
\def\l@app#1#2{\ifnum \c@tocdepth >\z@
    \addpenalty\@secpenalty
    \addvspace{1.0em \@plus\p@}%
    \setlength\@tempdima{8.5em}%
    \begingroup
      \parindent \z@ \rightskip \@pnumwidth
      \parfillskip -\@pnumwidth
      \leavevmode \bfseries
      \advance\leftskip\@tempdima
      \hskip -\leftskip
      #1\nobreak\hfil \nobreak\hb@xt@\@pnumwidth{\hss #2}\par
    \endgroup\fi}
\newcounter{sapp}[app]

\def\sapp{\def\theequation{A\arabic{app}.\arabic{equation}}\par
   \@afterindentfalse
  \secdef\@sapp\@dsapp}
\newcommand\@sapp{\@startsection{sapp}{2}{\z@}%
                                     {-3.25ex\@plus -1ex \@minus -.2ex}%
                                     {1.5ex \@plus .2ex}%
                                     {\normalfont\large\bfseries}}

\def\@dsapp#1{%
{\parindent \z@ \raggedright  \bf #1}\par\nobreak}
\newcommand{\l@sapp}{\@dottedtocline{2}{1.5em}{3em}}
\def\draft{\oddsidemargin -.5truein
        \def\@oddfoot{\sl preliminary draft \hfil
        \rm\thepage\hfil\sl\today\quad\militarytime}
        \let\@evenfoot\@oddfoot \overfullrule 3pt
        \let\label=\draftlabel
        \let\marginnote=\draftmarginnote
   \def\@eqnnum{(\theequation)\rlap{\kern\marginparsep\tt\@eqnlabel}%
\global\let\@eqnlabel\@vacuum}  }

\def\be{\begin{eqnarray}}
\def\ee{\end{eqnarray}}

\def\p{\partial}

\def\beq{\begin{equation}}
\def\eeq{\end{equation}}
\def\ba{\beq\new\begin{array}{c}}
\def\ea{\end{array}\eeq}
\def\be{\ba}
\def\ee{\ea}

\def\Tr{{\rm Tr}\,}

\def\diag{{\rm diag}\,}

\newfont{\Bbbb}{msbm7 scaled 1\@ptsize00}

\newcommand{\z}{\raise-1pt\hbox{$\mbox{\Bbbb Z}$}}
\def\normordboson{ {\scriptstyle {{*}\atop{*}}} }

\def\normord{ {\scriptstyle {{\bullet}\atop{\bullet}}} }

\newfont{\alef}{msbm10 at 11pt}
\newfont {\goth}{eufm10 at 11pt}
\def\mathbb#1{\hbox{{\alef #1}}}

\let\@@savethanks\thanks
\def\thanks#1{\gdef\thefootnote{\alph{footnote}}\@@savethanks{#1}}

\newtheorem{theorem}{Theorem}[section]
\newtheorem{lemma}[theorem]{Lemma}
\newtheorem{proposition}[theorem]{Proposition}

\newtheorem{conjecture}[theorem]{Conjecture}

\makeatletter
\g@addto@macro \normalsize {%
 \setlength\abovedisplayskip{14pt plus 3pt minus 3pt}%
 \setlength\belowdisplayskip{14pt plus 3pt minus 3pt}%
  \setlength\abovedisplayshortskip{11pt plus 3pt minus 3pt}%
 \setlength\belowdisplayshortskip{11pt plus 3pt minus 3pt}%
}
\makeatother

\bigskip
\bigskip

\title{
\bigskip
{\bf
Cut-and-join description of generalized Brezin--Gross--Witten model} \vspace{.5cm}}
\author{{\bf A. Alexandrov}\thanks{E-mail:  {\tt alexandrovsash at gmail.com}}
\date{ } \\
{\small {\it 
Center for Geometry and Physics, Institute for Basic Science (IBS), Pohang 37673, Korea
  \&}}\\
{\small {\it 
CRM,
Universit\'e de Montr\'eal, Montr\'eal,  Canada  \&}}\\
{\small {\it 
ITEP, Moscow, Russia}}\\
}

\begin{document}

\setcounter{footnote}{0}

\setcounter{tocdepth}{3}

\maketitle

\vspace{-8.0cm}

\begin{center}
\hfill ITEP/TH-18/16
\end{center}

\vspace{6.5cm}
%\bigskip
\begin{abstract} 
We investigate the Brezin--Gross--Witten model, a tau-function of the KdV hierarchy, and its natural one-parameter deformation, the generalized Brezin--Gross--Witten tau-function. In particular, we derive the Virasoro constraints, which completely specify the partition function. We solve them in terms of the cut-and-join operator. The Virasoro constraints lead to the loop equations,  which we solve in terms of the correlation functions. Explicit expressions for the coefficients of the tau-function and the free energy are derived, and a compact formula for the genus zero contribution is conjectured. A family of polynomial solutions of the KdV hierarchy, given by the Schur functions, is obtained for the half-integer values of the parameter. The quantum spectral curve and its classical limit are discussed. 
\end{abstract}
\bigskip


\bigskip

%{\small \bf MSC 2010 Primary: 37K10, 53D45, 81R10, 14N10; Secondary: 81T30.}

\begin{comment}
(	37K10  	Completely integrable systems, integrability tests, bi-Hamiltonian structures, hierarchies (KdV, KP, Toda, etc.)
	14N35  	Gromov-Witten invariants, quantum cohomology, Gopakumar-Vafa invariants, Donaldson-Thomas invariants
   53D45  	Gromov-Witten invariants, quantum cohomology, Frobenius manifolds
		81R10  	Infinite-dimensional groups and algebras motivated by physics, including Virasoro, Kac-Moody, $W$-algebras and other current algebras and their representations 
			81R12  	Relations with integrable systems 
				17B68  	Virasoro and related algebras
				22E70  	Applications of Lie groups to physics; explicit representations
				81T30  	String and superstring theories; other extended objects (e.g., branes))
14N10 Enumerative problems (combinatorial problems)
			\end{comment}

%\bigskip

%\bigskip

\newpage

\tableofcontents

\def\thefootnote{\arabic{footnote}}
\section{Introduction}
%\addcontentsline{toc}{section}{Introduction}
%\def\theequation{\arabic{equation}}
\setcounter{equation}{0}

The Brezin--Gross--Witten (BGW) model 
\be
Z_{BGW}=\int \left[d U\right] e^{\frac{1}{\hbar}\Tr ( A^\dagger U+A U^\dagger)}
\ee
was introduced in the lattice gauge theory over 35 years ago \cite{Brezin, GW}. Later it was shown that in the weak coupling phase this model satisfies the Virasoro constraints \cite{GN}. Moreover, it is a tau-function of the KdV integrable hierarchy and can be described by the generalized Kontsevich model \cite{MMS}. 

This makes the BGW model interesting and, in many respects, similar to the Kontsevich-Witten tau-function \cite{Konts,Witten} -- one of the most important and beautiful ingredients of the modern mathematical physics. However, unlike the Kontsevich--Witten (KW) tau-function, which generates the intersection numbers of the moduli spaces or Riemann surfaces, and many other matrix models, for which enumerative geometry/combinatorics interpretation is known, similar interpretation of the BGW tau-function is still not available. Using the generalized Kontsevich model description of this tau-function, one can try to identify it with the generating function of the $r$-spin intersection numbers for $r=-2$. However, corresponding geometrical construction is not available yet, thus, it is impossible to compare the intersection numbers with the correlation functions of the matrix model. 

In spite of this absence of geometrical interpretation, the BGW tau-function is known to play (similarly to the KW tau-function) an important role in the topological recursion/Givental decomposition \cite{Eyn1,Eyn2,Eyn3,Giv1,Giv2}. Namely, it appears in decomposition of the complex matrix model \cite{IMMM1,IMMM2,IMMM3} and, in general, corresponds to the hard walls (see \cite{ChekhovHW} and references therein). 

Recently, it was shown that a natural one parametric deformation of the KW tau-function, called the Kontsevich--Penner model, describes open intersection numbers\cite{Open1,Open2, Open0}, a new and extremely interesting set of enumerative geometry invariants, which was introduced in \cite{PST,Buryak}. The matrix integral description allows us to show that their generating function is a tau-function of the modified KP (MKP) hierarchy, and to construct a full family of the Virasoro and W-constraints. This model possess a number of nice properties and, arguably, is even more  beautiful and natural then any of its specifications (in particular, the KW tau-function). 

Thus, to find a natural interpretation of the BGW tau-function one can try to consider its deformation, analogous to the Kontsevich--Penner deformation of the KW tau-function. It is easy to construct this deformation using the generalized Kontsevich model representation. In this representation it corresponds to the logarithmic deformation of the potential. This deformed model was introduced in \cite{MMS} and is given by the matrix integral 
\be\label{intrf}
\tau_N\sim \displaystyle{\int \left[d \Phi\right] \exp\left( \Tr \left(\frac{\Lambda^2\Phi}{\hbar}+\frac{1}{\hbar\Phi} +(N-M) \log \Phi\right)\right)}.
\ee
From the general properties of the generalized Kontsevich model (GKM) \cite{GKMToda} it follows that it is a tau-function of the MKP hierarchy with discrete time $N$. However, other properties of this model have not been investigated in detail so far. In particular, the Virasoro constraints were not known. The main goal of this paper is to fill this gap and to describe the generalized BGW model (\ref{intrf}) and its interesting specifications, in particular the original BGW model. 

We show that the tau-function (\ref{intrf}) is well definite for any complex (not necessarily integer!) value of $N$. Moreover, for any given value of $N$ this is a tau-function of the KdV hierarchy. We describe the Kac--Schwarz algebra for this tau-function and derive the Virasoro constraints. Here the difference with the Kontsevich--Penner model is quite transparent: to describe the Kontsevich--Penner model one should introduce higher W-constraints, while the partition function of the generalized BGW model is completely fixed by the Virasoro constraints. Moreover, only the first of them (the string equation) depends on $N$, thus, on the level of linear constraints, the case of general $N$ is almost as simple as the case with $N=0$.

Often the Virasoro and W-constraints can be solved in terms of the cut-and-join operator. Corresponding method  was introduced in \cite{MSh} for the Gaussian branch of the Hermitian matrix model and later has been applied to the KW tau-function \cite{Kontscaj} and to the Kontsevich--Penner model \cite{Open1,Open2}. We solve the Virasoro constraints for the BGW and generalized BGW tau-functions in terms of the cut-and-join operator:
\be
\tau_N=e^{\hbar \widehat{W}_N}\cdot 1,
\ee
where
\be
\widehat{W}_{N}=\frac{1}{2}\sum_{k,m=0}^\infty (2k+1)(2m+1)t_{2k+1}t_{2m+1}\frac{\p}{\p t_{2k+2m+1}}\\
+\frac{1}{4}\sum_{k,m=0}^\infty(2k+2m+3)t_{2k+2m+3}\frac{\p^2}{\p t_{2k+1}\p t_{2m+1}}+\left(\frac{1}{16}-\frac{N^2}{4} \right)t_1.
\ee

Using this operator we derive the coefficients of expansion of the tau-function and free energy. Here we see that the case of generalized BGW tau-function is much more interesting comparing to the original BGW tau-function. In particular, while for the BGW tau-function the genus zero contribution to the free energy is equal to zero (and higher genera contributions are rational functions of only finite number of times), for general $N$ this is not the case. Namely, for any genus the free energy is a non-trivial function of all times. The results of computations  allow us to conjecture a compact expression for the genus zero free energy of the generalized BGW tau-function. 

We also derive an equation for the quantum spectral curve of the generalized BGW tau-function,
\be
\left(\hbar^2 x^2\frac{\p^2}{\p x^2}+\hbar^2 x\frac{\p}{\p x}-x-\frac{S^2}{4}\right) \Psi_S(x)=0,
\ee
where $S=\hbar^{-1}N$. As for other KP/Toda tau-functions, which describe the enumerative geometry invariants, the equation for the quantum spectral curve, up to a conjugation, coincides with one of the Kac--Schwarz operators\cite{Aenum,ALS,Open1}. In the classical limit we get a genus zero spectral curve with one branch point.

The Virasoro constraints allow us to derive the loop equations and to solve them recursively. The correlation functions are defined on the spectral curve and they are symmetric polynomials in the inverse global coordinate. Thus, corresponding differentials are meromorphic with poles only at the branch point.

For the half-integer values of the parameter $N$, the generalized BGW tau-function is a polynomial in times. More specifically, it is given by the Schur functions of the dilaton shifted times, labelled by the triangular partitions. We describe this family of the KdV tau-functions (which constitute an infinite MKP tau-function) in detail.

All this allows us to conclude that, as in the case of the Kontsevich--Penner model, the deformed model appears to be more beautiful and natural then the original one. Unfortunately, a unitary integral representation of this deformed model is not known, and we do not expect that this model is directly related to the original lattice gauge models. However, some of our results (in particular, the cut-and-join representation) should be useful for the original BGW model. Moreover, from the Virasoro constraints derived in Section \ref{Secvir} it follows that the generalized BGW model describes a model of the open-closed string theory involving gravity\cite{openstring}, which can be obtained from the unitary matrix model in a double scaling limit \cite{dsl}.

The present paper is organized as follows. In Section \ref{S1} we consider the original BGW model and, basically following \cite{MMS}, describe it in terms of the GKM.  Section \ref{GBGW} is devoted to the generalized BGW tau-function.  In the Appendices we present explicit expressions for expansion of the tau-function and free energy of BGW and generalized BGW tau-functions.

\section{Brezin--Gross--Witten model}\label{S1}

The partition function of the BGW model \cite{GW,Brezin} is given by an $M\times M$ unitary matrix integral
\be\label{unitm}
Z_{BGW}=\int \left[d U\right] e^{\frac{1}{\hbar}\Tr ( A^\dagger U+A U^\dagger)}.
\ee
Here the Haar measure on the unitary group $U(M)$ is normalised by $\int \left[d U\right]=1$ and the parameter $\hbar$ describes the topological expansion (see below). Naively, (\ref{unitm}) depends on two external matrices, $A$ and $A^\dagger$, but actually it depends only on their product, more precisely on the square root of it
\be
\Lambda:=\left(A^\dagger A\right)^{\frac{1}{2}}.
\ee
The behaviour of this matrix model is essentially different at large and small values of $\hbar^{-1}\Tr \Lambda^{-1}$ and there is a phase transition between these two regimes \cite{GW,Brezin,Wadia}. In this paper we consider only the so-called Kontsevich (weak coupling) phase, which corresponds to the large values of the eigenvalues of the matrix $\Lambda$. Below for simplicity we assume that the matrix $\Lambda$ is diagonal
\be
\Lambda=\diag(\lambda_1,\dots,\lambda_M).
\ee

\subsection{Description in terms of generalized Kontsevich model}

As many other important matrix models, the BGW model can be described in terms of the generalized Kontsevich model \cite{GKMToda}. Namely, as it was shown by  A.~Mironov, A.~Morozov and G.~W.~Semenoff in \cite{MMS},  
\be\label{BGWGKM}
Z_{BGW}=\frac{ \displaystyle{\int \left[d \Phi\right] \exp\left( \Tr \left(\frac{\Lambda^2\Phi}{\hbar}+\frac{1}{\hbar\Phi} -M \log \Phi\right)\right)}}{ \displaystyle{\int \left[d \Phi\right] \exp\left( \Tr \left( \frac{1}{\hbar\Phi}-M \log \Phi\right)\right)}}.
\ee
In this section we basically follow the approach of \cite{MMS}. 

Actually, (\ref{BGWGKM}) as well as (\ref{unitm}) depends only on the ratio $\Lambda/\hbar$, thus it is convenient to introduce
\be
\tilde{\Lambda}:=\frac{\Lambda}{\hbar}=\diag(\tilde{\lambda}_1,\dots,\tilde{\lambda}_M),
\ee
and $\tilde{\lambda}_i=\lambda_i/\hbar$.

In (\ref{BGWGKM}) we integrate over $M\times M$ normal matrices, that is diagonolizable matrices
\be
\Phi=U\, \diag(\phi_1,\dots,\phi_M) \,U^{\dagger},\,\,\,\,\,\,\,\,\,\, \phi_i \in \gamma,
\ee
where $U$ is unitary and the contour $\gamma$ runs from $-\infty$ to a small circle enclosing zero, and then returning to $-\infty$. Then the measure of integration can be expressed in terms of $U$ and $\phi_i$'s in the standard way
\be
\left[d \Phi\right] =\Delta (\phi)^2 \left[d U\right]\prod_{i=1}^M d \phi_i,
\ee 
where
\be
\Delta(\phi)=\prod_{i<j}(\phi_j-\phi_i)
\ee
is the Vandermonde determinant.

After integration over the unitary matrix $U$ with the help of the HCIZ formula, (\ref{BGWGKM}) reduces to
\be\label{ZBGW}
Z_{BGW}=(-1)^\frac{M(M-1)}{2}\, \prod_{j=1}^{M}(j-1)! \, \frac{\det_{i,j=1}^M\left(\tilde{\lambda}_j^{M-i}\,I_{M-i}(2\tilde{\lambda}_j)\right)}{\Delta(\tilde{\lambda}^2)}.
\ee
Here
\be\label{mbf}
I_\nu(x)=\left(\frac{2}{x}\right)^\nu \frac{1}{2\pi i} \int_{\gamma} e^{\frac{x^2 \phi}{4}+\frac{1}{\phi}} \frac{d\phi}{\phi^{\nu+1}}
\ee
is the modified Bessel function and the normalization of (\ref{ZBGW}) can be easily found from its small $x$ expansion
\be
I_\nu(x)=\frac{1}{\Gamma(\nu+1)}\left(\frac{x}{2}\right)^\nu(1+O(x)).
\ee

From this eigenvalue integral representation it immediately follows that in the Kontsevich phase
\be\label{taufun}
\tau_{BGW}(\Lambda)={\mathcal C}_{BGW}^{-1}\,Z_{BGW},
\ee
where
\be
{\mathcal C}_{BGW}=\frac{e^{2\Tr \tilde{\Lambda}}\prod_{i=1}^{M}\left(j-1\right)!}{(2\pi)^{\frac{M}{2}}\det\left(\tilde{\Lambda}\otimes 1+1\otimes \tilde{\Lambda}\right)^\frac{1}{2}},
\ee
is a tau-function of the KP hierarchy. Indeed,
\be\label{BGWdet}
\tau_{BGW}(\Lambda)=\frac{\det_{i,j=1}^M{\Phi_j(\lambda_i)}}{\Delta(\lambda)},
\ee
which defines a tau-function in the Miwa parametrization 
\be\label{Miwa}
t_k=\frac{1}{k}\Tr \Lambda^{-k}.
\ee
Here $\Phi_j$'s are the so called basis vectors, which can be expressed in terms of the modified Bessel functions (\ref{mbf}),
\be\label{Phi0}
\Phi_j(\lambda)=\sqrt{4\pi \tilde{\lambda}}\, \lambda^{j-1} e^{-2\tilde{\lambda}} I_{j-1}(2\tilde{\lambda})\\
= \frac{\sqrt{4\pi \tilde{\lambda}}}{2\pi i} \hbar^{j-1} e^{-2\tilde{\lambda}} \int_{\gamma} e^{\tilde{\lambda}^2 t+\frac{1}{t}} \frac{dt}{t^{j}}.
\ee
We consider only the asymptotic expansion of the modified Bessel function for large values of $\lambda$ (we assume that $\arg \lambda \neq \pi$)
\be\label{Phi1ex}
\Phi_j(\lambda)=\lambda^{j-1}\left(1+\sum_{k=1}^\infty\frac{(-\hbar)^k}{\lambda^k}\frac{a_k(j)}{16^k\, k!}\right),
\ee
where
\be\label{ak}
a_k(j)=(4(j-1)^2-1^2)(4(j-1)^2-3^2)\dots(4(j-1)^2-(2k-1)^2),
\ee
thus, $\Phi_j(\lambda)$'s are of the form 
\be\label{Phias}
\Phi_j(\lambda)=\lambda^{j-1}(1+O(\lambda^{-1})).
\ee
This guarantees that 
\be\label{norm}
\tau_{BGW}(\Lambda)=1+O(\lambda_j^{-1}).
\ee

Vectors (\ref{Phi0}) are defined for all $j\in \mathbb Z$. Vectors for $j\geq 1$ define a point of the big cell of the Sato Grassmannian \cite{Sato, Segal,Mulase}\footnote{In this paper we consider only the index (or charge) zero sector of the Sato Grassmannian, thus all 
points corresponding to the different values of the discrete time are described in the same space. Equivalent description should include a flag of the Sato Grassmannians with different indices.}
\be\label{BGWpoint}
{\mathcal W}_{BGW}=\left<\Phi_1,\Phi_2,\Phi_3,\dots\right>.
\ee
Any such point corresponds to a tau-function of the KP hierarchy, which is a formal series in the times $t_k$ and solves the bilinear identity
\begin{equation}
\oint_{{\infty}} e^{\xi ({\bf t}-{\bf t'},z)}
\,\tau ({\bf t}-[z^{-1}],\hbar)\,\tau ({\bf t'}+[z^{-1}],\hbar)dz =0.
\end{equation}
Here $\xi({\bf t},z)=\sum_{k=1}^\infty t_k z^{k}$ and we use the standard notation
\be\label{shiftedt}
{\bf t} \pm \left[z^{-1}\right]=\left\{t_1\pm\frac{1}{z},t_2\pm\frac{1}{2z^2},t_3\pm\frac{1}{3z^3},\dots\right\}.
\ee
Thus, the BGW tau-function
\be
\tau_{BGW}({\bf t},\hbar)
\ee
 is defined by the point
(\ref{BGWpoint}), or equivalently, it can be considered as a limit of the ration of determianants (\ref{BGWdet}) as the size of the matrices $M$ tends to infinity. In this limit all the Miwa variables (\ref{Miwa}) are independent. 

In the Sato Grassmannian description the first basis vector plays a special role. It is related to the tau-function by 
\be\label{BA}
\Phi_1(\lambda)=\tau ([\lambda^{-1}],\hbar),
\ee
and is equal to the dual Baker--Akhiezer function at ${\bf t}=0$.

It is clear that the parameter $\hbar$ is not independent and can be removed by the time variables rescaling
\be\label{Homo}
\tau_{BGW}({\bf t},\hbar)=\tau_{BGW}({\bf t},1 )\Big|_{t_k=\hbar^k t_k}.
\ee
Let us stress that the expansion of $\tau_{BGW}({\bf t},\hbar)$ in $\hbar$ is not the genus expansion, but the topological expansion. More concretely,
\be
\tau_{BGW}({\bf t},\hbar) =\exp\left(\sum_{g=0}^\infty\sum_{n=1}^\infty \hbar^{-\chi}{\mathcal F}_{g,n}({\bf t})\right),
\ee
where $\chi=2-2g-n$ can be considered as the Euler characteristic. Here ${\mathcal F}_{g,n}({\bf t})$ is a genus $g$ contribution to free energy, which is a homogeneous polynomial in times $t_k$ of degree $n$,
\be
\sum_{k=0}^\infty t_{k}\frac{\p}{\p t_{k}}{\mathcal F}_{g,n}({\bf t})=n\,{\mathcal F}_{g,n}({\bf t}).
\ee
To get the genus expansion, one should multiply the times  by $\hbar^{-1}$:
\be\label{genusex}
\tau_{BGW}(\hbar^{-1}{\bf t},\hbar) =\exp\left(\sum_{g=0}^\infty \hbar^{2g-2}{\mathcal F}_{g}({\bf t})\right).
\ee
${\mathcal F}_{g}({\bf t})$ is the genus $g$ contribution to the free energy and
\be
{\mathcal F}_{g}({\bf t})=\sum_{n=1}^\infty {\mathcal F}_{g,n}({\bf t}).
\ee
It is known \cite{GN1,IMMM3} that
\be
{\mathcal F}_{0}=0,\\
{\mathcal F}_{1}=-\frac{1}{8}\log\left(1-\frac{t_1}{2}\right),
\ee
and for $g>1$ all ${\mathcal F}_g$ are polynomials in the variables
\be
T_k=\frac{t_k}{(2-t_1)^k}.
\ee
Variables $T_k$ are the ``moment variables" and expressions for ${\mathcal F}_{k}({\bf T})$ for small $k$ were obtained in \cite{GN1,IMMM3}.
With the help of the cut-and-join description of Section \ref{CAJs} we are able to find expressions for  ${\mathcal F}_{g}({\bf T})$ for $g\leq30$. See Appendix \ref{A} for the expressions of ${\mathcal F}_{g}({\bf T})$ for $g\leq9$.

\subsection{KdV hierarchy and Virasoro constraints}\label{KdVsec}

It is well-known that the tau-function $\tau_{BGW}({\bf t},\hbar)$ does not depend on even times $t_{2k}$ \cite{GN1}. Thus, it is a tau-function of the 2-reduction of the KP hierarchy, which is the KdV hierarchy \cite{MMS}. 
Probably the simplest way to show it is to use the Sato Grassmannian description and the Kac--Schwarz  operators \cite{KS} as it was done in \cite{MMS}.

The Kac--Schwarz (KS) operators \cite{KS,Fukuma,Aenum,GKM,AdlervM,MMS} are the differential operators in one variable which stabilize the point of the Sato Grassmannian for a given tau-function. For any tau-function the corresponding KS operators constitute an algebra (a subalgebra in $w_{1+\infty}$). Thus, for any KS operator we can use a correspondence between the $w_{1+\infty}$ and $W_{1+\infty}$ algebras \cite{KS,Fukuma,AdlervM, Orlov} to construct an operator from $W_{1+\infty}$, which annihilates the tau-function. 

Let us consider the operators
\be\label{KSop}
a=\frac{\lambda}{2}\frac{\p}{\p \lambda}+\frac{\lambda}{\hbar}-\frac{1}{4},\\
b=\lambda^2,
\ee
satisfying the commutation relations
\be\label{abcom}
\left[a,b\right]=b.
\ee
Using the integral representation (\ref{Phi0}) of the basis vectors it is easy to show \cite{MMS} that
\be\label{KS0}
a\, \Phi_j= (j-1)\Phi_j+\frac{1}{\hbar}\Phi_{j+1},\\
b\,  \Phi_j =j \hbar \Phi_{j+1}+\Phi_{j+2},
\ee
thus operators $a$ and $b$ stabilize the point (\ref{BGWpoint}) of the Sato Grassmannian 
\be
a \,{\mathcal W}_{BGW} \subset  {\mathcal W}_{BGW},\\
b\, {\mathcal W}_{BGW} \subset  {\mathcal W}_{BGW},
\ee
and are the KS operators.

However, these two operators do not completely specify the point of the Sato Grassmannian and the tau-function. Thus, they do not generate the KS algebra. Let us find some other KS operators. Integration by parts yields
\be\label{almKS}
\frac{1}{b}a\, \Phi_j= \left(\frac{1}{2\lambda}\frac{\p}{\p \lambda}+\frac{1}{\hbar \lambda}-\frac{1}{4\lambda^2}\right)\Phi_j=\frac{1}{\hbar}\Phi_{j-1}.
\ee
The operator $\frac{1}{b}a$ is not a KS operator
\be
\frac{1}{b}a\, \Phi_1 =\frac{1}{\hbar}\Phi_0 \notin {\mathcal W}_{BGW}.
\ee
However, combining (\ref{almKS}) with (\ref{KS0}) one obtains
\be
\frac{1}{b}a^2 \, \Phi_j = \frac{1}{\hbar}(j-1)\Phi_{j-1}+\frac{1}{\hbar^2}\Phi_{j} 
\ee
and
\be\label{cexp}
c=\frac{1}{b}a^2=\frac{1}{4}\frac{\p^2}{\p \lambda^2}+\frac{1}{\hbar}\frac{\p}{\p \lambda}+ \frac{1}{\hbar^2}+\frac{1}{16\lambda^2} 
\ee
is the KS operator. To the best of our knowledge, this KS operator for the BGW tau-function has never been considered . Operators $a$, $b$ and $c$ satisfy the commutation relations
\be
\left[c,a\right]=c,\,\,\,\,\,\,\left[c,b\right]=2a+1,
\ee
and (\ref{abcom}).

\begin{proposition}
Operators $a$ and $c$ completely specify the point ${\mathcal W}_{BGW}$ of the Sato Grassmannian. 
\end{proposition}

\begin{proof}
From (\ref{cexp}) we see that the operator $c$ acts as
\be
c\, \lambda^k=\frac{1}{\hbar^2}\lambda^k\left(1+O(\lambda^{-1})\right).
\ee 
Thus, if this is the KS operator for some point of the Sato Grassmannian, then the first basis vector should be the eigenfunction of this operator:
\be
c\, \Phi_1 =\frac{1}{\hbar^2} \Phi_1.
\ee
From this equation it immediately follows that the solution corresponds to the big cell of the Sato Grassmannian,
\be
\Phi_1=1+O(\lambda^{-1}),
\ee
and it is unique. All higher basis vectors can be generated from $\Phi_1$ by the operator $a$.
\end{proof}

From the correspondence between $w_{1+\infty}$ and its central extension $W_{1+\infty}$ 
it immediately follows that the KS operators $b^{k}$ and $b^{k}a$ correspond to the constraints
\be\label{KdVconst}
\frac{\p}{\p t_{2k}} \tau_{BGW}=\nu_k \, \tau_{BGW},\,\,\,\,\, k\geq 1,
\ee
and
\be\label{KPvir}
\left(\frac{1}{2}\widehat{L}_{2k}-\frac{1}{\hbar}\frac{\p}{\p t_{2k+1}}\right) \tau_{BGW} = \mu_k \, \tau_{BGW},\,\,\,\,k\geq0,
\ee
for some constants $\nu_k$ and $\mu_k$. Here
\be
\widehat{L}_m=\frac{1}{2} \sum_{a+b=-m}a b t_a t_b+ \sum_{k=1}^\infty k t_k \frac{\p}{\p t_{k+m}}+\frac{1}{2} \sum_{a+b=m} \frac{\p^2}{\p t_a \p t_b}
\ee
is an operator from the Virasoro subalgebra of the $W_{1+\infty}$ symmetry algebra of the KP hierarchy.

From the commutation relations between the operators in the l.h.s. of (\ref{KdVconst}) and (\ref{KPvir}) it follows that 
\be
\nu_k=\mu_k=0,\,\,\,\,\,\,\,\,\,\,\,\,\,\,\,k>0.
\ee
However, this argument does not allow us to find $\mu_0$. This fact corresponds to the observation that the KS operators $a$ and $b$ do not completely specify a point of the Sato Grassmannian. From the normalization condition (\ref{norm}) and the constraint (\ref{KPvir}) with $k=0$ it follows that this constant is proportional to the first derivative of the tau-function:
\be
\mu_0=-\frac{1}{\hbar}\frac{\p}{\p t_1} \tau_{BGW}\Big|_{{\bf t}=0}.
\ee
This derivative is equal to the coefficient in front of $\lambda^{-1}$ of the expansion (\ref{Phi1ex}) of $\Phi_1(\lambda)$,
\be
\Phi_1(\lambda)=1+\frac{\hbar}{16\lambda}+O(\lambda^{-2}),
\ee
thus
\be
\mu_0=-\frac{1}{16}.
\ee

Since the tau-function is independent of the even times, the Virasoro constraints (\ref{KdVconst}) can be represented as
\be\label{KdVvir}
\hbar \widehat{\mathcal L}_m\, \tau_{BGW}({\bf t},\hbar)= \frac{\p}{\p t_{2m+1}}\tau_{BGW}({\bf t},\hbar), \,\,\,\,\, m\geq 0,
\ee
where
\be
\widehat{\mathcal L}_m:= \frac{1}{2}\sum_{k=0}^\infty (2k+1) {t}_{2k+1} \frac{\p}{\p t_{2k+2m+1}}+\frac{1}{4} \sum_{a+b=m-1} \frac{\p^2}{\p t_{2a+1} \p t_{2b+1}}+\frac{1}{16}\delta_{m,0}.
\ee
%and times are subject to the dilaton shift
%\be
%\tilde{t}_k=t_k-2\frac{ \delta_{k,1}}{\hbar}.
%\ee
These Virasoro constraints for the BGW tau-function were obtained already in \cite{GN}. 
Constraints (\ref{KdVvir}) have a unique solution with the normalisation (\ref{norm}). This solution will be constructed in the next section.

The KS operator $c$
corresponds to the $W_{1+\infty}$ operator
\be
\widehat{W}_c=\frac{1}{4}\widehat{M}_{-2}+\frac{1}{\hbar}\widehat{L}_{-1}-\frac{1}{8}t_2,
\ee
where
\begin{multline}
\widehat{M}_k=\frac{1}{3} \sum_{a+b+c=k} \normordboson \widehat{J}_a \widehat{J}_b \widehat{J}_c \normordboson=
\frac{1}{3}\sum_{a+b+c=-k}a\, b\, c\, t_a\, t_b\, t_c+\sum_{c-a-b=k}a\, b\, t_a\, t_b\, \frac{\p}{\p t_{c}}\\
+\sum_{b+c-a=k}a\, t_{a}\frac{\p^2}{\p t_b\p t_c}+\frac{1}{3}\sum_{a+b+c=k}\frac{\p^3}{\p t_a \p t_b \p t_c}
\end{multline}
are the cubic operators from the $W_{1+\infty}$ algebra. Thus, $\tau_{BGW}$ is the eigenfunction of the operator $\widehat{W}_c$ and, from the consideration of the corresponding linear constraint at the point $t_k=0$ for all $k$ we conclude that the eigenvalue is equal to zero:
\be\label{Wc}
\widehat{W}_c\, \tau_{BGW}=0.
\ee
This equation also allows us to find $\mu_0$. Indeed, from the KdV reduction condition (\ref{KdVconst}) it follows that (\ref{Wc}) is equivalent to
\be
\sum_{k=0}^\infty (2k+2)t_{2k+2} \left(\widehat{\mathcal L}_k-\frac{1}{\hbar}\frac{\p}{\p t_{2k+1}}\right)\tau_{BGW}=0.
\ee

\subsection{Cut-and-join operator}\label{CAJs}

Using the approach introduced in \cite{MSh} we solve the constraints (\ref{KdVvir}) and 
construct a simple recursion, which allows us to calculate the coefficients of the $\hbar$-expansion of the tau-function
\be\label{hbarexp}
\tau_{BGW}({\bf t},\hbar)=1+\sum_{k=1}^\infty \hbar^k \tau^{(k)}_{BGW}({\bf t}).
\ee
Namely, we introduce the Euler operator
\be
\widehat{D}:=\sum_{k=0}^\infty(2k+1)t_{2k+1}\frac{\p}{\p t_{2k+1}}.
\ee
Then, combining the Virasoro constraints (\ref{KdVvir}) we obtain
\be\label{Wequ}
\hbar \widehat{W}_{BGW}\,\tau_{BGW}=\widehat{D}\,\tau_{BGW},
\ee
where
\be\label{caj}
\widehat{W}_{BGW}=\sum_{k=0}^\infty (2k+1)t_{2k+1} \widehat{\mathcal L}_k\\
=\frac{1}{2}\sum_{k,m=0}^\infty (2k+1)(2m+1)t_{2k+1}t_{2m+1}\frac{\p}{\p t_{2k+2m+1}}\\
+\frac{1}{4}\sum_{k,m=0}^\infty(2k+2m+3)t_{2k+2m+3}\frac{\p^2}{\p t_{2k+1}\p t_{2m+1}}+\frac{t_1}{16}.
\ee
does not depend on $\hbar$. From (\ref{Homo}) it follows that
\be
\widehat{D} \, \tau^{(k)}_{BGW}= k\, \tau^{(k)}_{BGW}.
\ee
and after substitution of (\ref{hbarexp}) into (\ref{Wequ}) we get a recursion
\be\label{rec0}
 \tau^{(k+1)}_{BGW}=\frac{1}{k+1}\widehat{W}_{BGW}\, \tau^{(k)}_{BGW}.
\ee
Since $\tau_{BGW}^{(0)}=1$, we have
\be
\tau^{(k)}_{BGW}= \frac{\widehat{W}_{BGW}^k}{k!}\cdot 1.
\ee
Thus, we proved 
\begin{theorem}
\be
\tau_{BGW}= e^{\hbar \widehat{W}_{BGW}} \cdot 1
\ee
where the differential operator $\widehat{W}_{BGW}$ is given by (\ref{caj}).
\end{theorem}

With a few lines of Maple code the author was able to find all $\tau^{(k)}_{BGW}$ for $k\leq 90$. Let us stress that the obtained expressions allow us to find  explicitly all correlation functions $\omega_{g,n}$ for $g\leq 30$ and arbitrary $n$ (see below).

\section{Generalized Brezin--Gross--Witten model}\label{GBGW}

There exists a deformation of the BGW model, which depends on an additional parameter $N$ (not to be confused with $M$, the size of the matrices)
\be\label{gend}
Z_{N}(\Lambda)=\frac{ \displaystyle{\int \left[d \Phi\right] \exp\left( \Tr \left(\frac{\Lambda^2\Phi}{\hbar}+\frac{1}{\hbar\Phi} +(N-M) \log \Phi\right)\right)}}{ \displaystyle{\int \left[d \Phi\right] \exp\left( \Tr \left( \frac{1}{\hbar\Phi}+(N-M) \log \Phi\right)\right)}}.
\ee
For $N=0$ it obviously coincides with the BGW model (\ref{BGWGKM}), and for
$N\neq 0$ the unitary integral representation of (\ref{gend}) is not known.

This model was introduced in \cite{MMS}, and in the weak coupling limit (large $\tilde{\Lambda}$) it has very natural integrable properties.  Namely, from the general theory of GKM \cite{GKMToda}, it follows that after a multiplication by a simple quasi-classical prefactor it is a tau-function of the MKP hierarchy, where $N\in\mathbb Z$ is the discrete time. 

Following the description of the open intersection numbers in terms of the 
Kontsevich--Penner model, we do not require $N$ to be an integer. It appears that the model (\ref{gend}) is defined perfectly well for an arbitrary $N\in {\mathbb C}$. Moreover, the tau-functions corresponding to the half-integer values of $N$ are particularly interesting: they are polynomials. We call (\ref{gend}) the generalized Brezin--Gross--Witten model. In this section we consider the generalized BGW tau-function in detail.

\subsection{MKP hierarchy and Virasoro constraints}\label{Secvir}
After integration over the unitary group (\ref{gend}) reduces to
\be\label{Zdet}
Z_N(\Lambda)=(-1)^\frac{M(M-1)}{2}\,\det(\tilde{\Lambda})^{2N}\, \prod_{j=1}^{M}\Gamma(j-N)  \frac{\det_{i,j=1}^M\left(\tilde{\lambda}_j^{M-N-i}\,I_{M-N-i}(2\tilde{\lambda}_j)\right)}{\Delta(\tilde{\lambda}^2)},
\ee
which satisfies $Z_N(0)=1$.

From the general theory of GKM it follows that for the large values of the eigenvalues of $\Lambda$ matrix integral (\ref{gend}) corresponds to the MKP tau-function
\be\label{taufunG}
\tau_N={\mathcal C}_{N}^{-1}\,Z_{N},
\ee
where
\be\label{corm}
{\mathcal C}_{N}=\frac{e^{2\Tr \tilde{\Lambda}}\det \tilde{\Lambda}^N\,\prod_{i=1}^{M}\Gamma(j-N) }{(2\pi)^{\frac{M}{2}}\,\det\left(\tilde{\Lambda}\otimes 1+1\otimes \tilde{\Lambda}\right)^\frac{1}{2}}.
\ee
Indeed, from (\ref{Zdet}) and (\ref{corm}) we have
\be\label{BGWN}
\tau_N=\frac{\det_{i,j=1}^M{\left(\Phi^{(N)}_{j}(\lambda_i)\right)}}{\Delta(\lambda)},
\ee
where the basis vectors
\be\label{PhiN}
\Phi_j^{(N)}(\lambda):=\lambda^N \Phi_{j-N}(\lambda),
\ee
and $\Phi_j$'s were defined in (\ref{Phi0}).
The coefficients of their asymptotic series expansion for the large values of $|\lambda|$ depend only on $j-N$
\be\label{phiser}
\Phi^{(N)}_j(\lambda)=\lambda^{j-1}\left(1+\sum_{k=1}^\infty\frac{(-\hbar)^k}{\lambda^k}\frac{a_k(j-N)}{16^k\, k!}\right),
\ee
where $a_k(j)$ is a polynomial both in $k$ and $j$ given by (\ref{ak}).
These basis vectors define a point on the big cell of the Sato Grassmannian
\be
{\mathcal W}_N=\left<\Phi^{(N)}_1,\Phi^{(N)}_2,\Phi^{(N)}_3,\dots\right>.
\ea
The value $N=0$ corresponds to the original BGW model considered in Section \ref{S1}:
\be
\tau_0=\tau_{BGW}.
\ee
From (\ref{phiser}) it follows that
\be\label{HomoN}
\tau_{N}({\bf t},\hbar)=\tau_{N}({\bf t},1 )\Big|_{t_k=\hbar^k t_k}.
\ee

Let us stress that (\ref{phiser}) defines a point of the big cell of the Sato Grassmannian, thus, a KP tau-function for any $N\in\mathbb C$. Moreover, it defines an MKP hierarchy, which relates $\tau_N$ and $\tau_{N+n}$ for any $n \in \mathbb Z$, $N\in\mathbb C$.
The MKP hierarchy can be described by the bilinear identity, satisfied by the tau-function $\tau_N({\bf t},\hbar)$, namely, in our case, 
\begin{equation}\label{bi1}
\oint_{{\infty}} z^n e^{\xi ({\bf t}-{\bf t'},z)}
\,\tau_{N+n} ({\bf t}-[z^{-1}],\hbar)\,\tau_{N} ({\bf t'}+[z^{-1}],\hbar)dz =0,\,\,\,\,N\in {\mathbb C}, \,\,\,\,n \in {\mathbb N}_0.
\end{equation}
Here ${\mathbb N}_0=\left\{0,1,2,\dots\right\}$ is the set of non-negative integers.

Again, for all $N$ we have\footnote{This expression for the KS operators indicates that the generalized BGW tau-function is closely related to the model, considered in \cite{AdlerMor}.}
\be\label{KS}
a\, \Phi_j^{(N)}=\left(j-1-\frac{N}{2}\right)\Phi_j^{(N)}+\frac{1}{\hbar}\Phi_{j+1}^{(N)},\\
b\,  \Phi_j^{(N)} =(j-N) \hbar \Phi_{j+1}^{(N)}+\Phi_{j+2}^{(N)}.
\ee
Here the KS operators $a$ and $b$ are given by (\ref{KSop}) and do not depend on $N$. This means, in particular, that they can not uniquely specify the point of the Sato Grassmannian, because they stabilize all points ${\mathcal W}_N$.

Integration by parts yields
\be
\frac{1}{b}\left(a-\frac{N}{2}\right)  \Phi_j^{(N)}=\frac{1}{\hbar} \Phi_{j-1}^{(N)}.
\ee
Thus
\be
c_N=\frac{1}{b}\left(a^2-\frac{N^2}{4}\right)
\ee
is the KS operator for $\tau_N$:
\be\label{protoQSC}
c_N\, \Phi_j^{(N)}=\frac{1}{\hbar}(j-1)\Phi_{j-1}^{(N)}+\frac{1}{\hbar^2}\Phi_{j}^{(N)}
\ee
and
\be
c_N\, {\mathcal W}_{N}= {\mathcal W}_{N}.
\ee
It satisfies the commutation relations
\be
\left[c_N,a\right]=c_N,\,\,\,\,\,\,\left[c_N,b\right]=2a+1,
\ee
and, similar to the case $N=0$ considered in Section \ref{S1}, we have
\begin{proposition}
Operators $a$ and $c_N$ completely specify the point ${\mathcal W}_{N}$ of the Sato Grassmannian. 
\end{proposition}

Using the Kac--Schwarz description (\ref{KS}) it is easy to show that the tau-function $\tau_N({\bf t},\hbar)$ satisfies the Virasoro constraints
\be\label{KdVvirN}
\hbar \widehat{\mathcal L}_m^{(N)}\, \tau_N({\bf t},\hbar)= \frac{\p}{\p t_{2m+1}}\tau_N({\bf t},\hbar), \,\,\,\,\, m\geq 0,
\ee
where
\be
\widehat{\mathcal L}^{(N)}_m= \frac{1}{2}\sum_{k=0}^\infty (2k+1) {t}_{2k+1} \frac{\p}{\p t_{2k+2m+1}}+\frac{1}{4} \sum_{a+b=m-1} \frac{\p^2}{\p t_{2a+1} \p t_{2b+1}}+\mu_0\delta_{m,0},
\ee
and 
\be
\mu_0=\frac{1}{16}-\frac{N^2}{4}.
\ee  
Again, the value of $\mu_0$ can be extracted from the expansion of the first basis vector
\be
\Phi^{(N)}_1(\lambda)=1+\hbar \frac{1-4N^2}{16\lambda}+O(\lambda^{-2}).
\ee
In the next section we prove
\begin{theorem}
There exists a unique (up to normalization) solution of the Virasoro constraints (\ref{KdVvirN}).
\end{theorem}
This theorem for $\hbar=1$ was proved in \cite{Horozov}, we prove it constructively and describe the solution in terms of the cut-and-join operator. 
%The partition function of the generalized BGW model (\ref{taufunG}) is a tau-function of the KdV hierarchy for any $N\in {\mathbb C}$. This tau-function satisfies the Virasoro constraints
Thus, the generalized BGW tau-function $\tau_N({\bf t},\hbar)$ is the unique solution of the Virasoro constraints (\ref{KdVvirN})  which satisfies the normalisation condition
\be\label{normc}
\tau_N({\bf 0},\hbar)=1.
\ee

Equation (\ref{KdVvirN}) for $m=0$ is the string equation for the generalized BGW tau-function. From the KS description it follows that this equation completely specifies the KdV tau-function.
\begin{lemma}
There is only one tau-function of the KdV hierarchy, which satisfies the string equation
\be
\hbar \widehat{\mathcal L}_0^{(N)}\, \tau_N({\bf t})= \frac{\p}{\p t_{1}}\tau_N({\bf t})
\ee
and the normalization condition (\ref{normc}).
\end{lemma}

Alternatively, the Virasoro constraints can be derived from the expansion of the operator $\widehat{W}_{c_N}$, the derivation is completely similar to the one from Section \ref{KdVsec}.  In particular, this operator specifies the value of constant $\mu_0$.

\subsection{Cut-and-join operator}

Similar to the case of the BGW tau-function, considered in Section \ref{S1}, we can solve the Virasoro constraints for the generalized BGW tau-function in terms of the cut-and-join operator:
\begin{lemma}
The solution of the Virasoro constraints (\ref{KdVvirN}) with the normalization (\ref{normc}) is given by
\be\label{Gencaj}
\tau_N({\bf t})=e^{\hbar \widehat{W}_{N}}\cdot 1
\ee
where
\be\label{WN}
\widehat{W}_{N}=\frac{1}{2}\sum_{k,m=0}^\infty (2k+1)(2m+1)t_{2k+1}t_{2m+1}\frac{\p}{\p t_{2k+2m+1}}\\
+\frac{1}{4}\sum_{k,m=0}^\infty(2k+2m+3)t_{2k+2m+3}\frac{\p^2}{\p t_{2k+1}\p t_{2m+1}}+\left(\frac{1}{16}-\frac{N^2}{4} \right)t_1\\
=\widehat{W}_{BGW}-\frac{N^2}{4}t_1.
\ee
All other solutions of the Virasoro constraints (\ref{KdVvirN}) correspond to the multiplication of (\ref{Gencaj}) by a constant.
\end{lemma}
\begin{proof}
Let us consider an arbitrary series in the time variables $t_k$ 
\be
Z({\bf t})=C+\sum_{k=1}^{\infty} Z^{(k)}({\bf t}).
\ee
where $Z^{(k)}({\bf t})$ is a homogeneous polynomial of degree $k$,
\be\label{neqf}
\widehat{D} Z^{(k)}({\bf t})= k\,Z^{(k)}({\bf t}),
\ee
and $C$ is some constant.
Then, if $Z(t)$ solves the Virasoro constraints (\ref{KdVvirN}), then
\be
\hbar \widehat{W}_{N}\,Z({\bf t})= \widehat{D}  Z({\bf t}).
\ee
From the comparison of the terms in the r.h.s. and the r.h.s. with the same degree we conclude
\be
\hbar \widehat{W}_{N}\,Z^{(k)}({\bf t})=\widehat{D} Z^{(k+1)}({\bf t}),
\ee
thus, from (\ref{neqf}) it follows that
\be
Z^{(k+1)}({\bf t})=\frac{\hbar}{k+1} \widehat{W}_{N}\,Z^{(k)}({\bf t})
\ee
or
\be
Z^{(k)}({\bf t})=\frac{\hbar^{k}}{k!}\widehat{W}_{N}^k\, C.
\ee
In particular, for $C=1$ we get the solution (\ref{Gencaj}), which coincides with the generalized BGW tau-function.
\end{proof}

We call (\ref{WN}) the cut-and-join operator for the generalized BGW tau-function.
This operator does not belong to the $W_{1+\infty}$ algebra of symmetries of the KP hierarchy, thus, integrability is not obvious from the representation (\ref{Gencaj}).

From the proof of Lemma \ref{WN} we see that the coefficients of the topological expansion
\be\label{sumhbar}
\tau_N({\bf t},\hbar)=1+\sum_{k=1}^\infty \hbar^k\, \tau_{N}^{ {(k)}},
\ee
satisfy the recursion
\be\label{recurN}
 \tau^{(k+1)}_{N}=\frac{1}{k+1}\widehat{W}_{N}\, \tau^{(k)}_{N}.
\ee
Using this recursion we calculated $ \tau^{(k)}_{N}$ for $k\leq60$, expressions for $k\leq 10$ are given in Appendix \ref{B}. There we introduce
\be
B_k(N)=(-1)^k a_k(N+1)=(1-4N^2)(3^2-4N^2)\dots((2k-1)^2-4N^2).
\ee
In Section \ref{Polysec} we show that  for $k\leq \frac{m(m+1)}{2}$ the polynomials $\tau_{N}^{(k)}$ are divisible by $B_m(N)$.

From (\ref{Gencaj}) we see that the tau-function is actually a series in $N^2$ (not in $N$), thus 
\be\label{Nsym}
\tau_{-N}^{BGW}({\bf t},\hbar)=\tau_{N}^{BGW}({\bf t},\hbar).
\ee
 From this observation and from the explicit expression for $\tau_N^{(1)}$ we conclude
\begin{lemma}
\be
\tau_{N}({\bf t},\hbar)=\tau_{\tilde{N}}({\bf t},\hbar)
\ee
if and inly if $\tilde{N}=\pm N$.
\end{lemma}
In particular, it means that the generalized BGW tau-function is not periodic in the variable $N$, and it is enough to consider only the values of $N$ with $\Re \, N\geq 0$.

Operator (\ref{caj}) has a rather natural free field representation. Indeed, let us introduce a bosonic current
\be
\widehat{J}(z)=\sum_{k=1}^\infty \left((2k+1)t_{2k+1}z^{2k}+\frac{1}{2z^{2k+2}}\frac{\p}{\p t_{2k+1}}\right)+\frac{i N}{2z},
\ee
then
\be
\widehat{W}_{N}=\frac{1}{2\pi i}\oint \left(\normord\frac{1}{3}\widehat{J}(z)^3+\frac{1}{16z^2}\widehat{J}(z)\normord\right) z dz,
\ee
where we use the standard normal ordering for the bosonic operators.

To consider the genus expansion we have to rescale the times $t_k\mapsto \hbar^{-1} t_k$. Then we can rewrite (\ref{Gencaj}) as
\be\label{BCH}
\tau_N(\hbar^{-1}{\bf t},\hbar)=e^{\frac{1}{\hbar^2}\widehat{W}^{(-1)}+\widehat{W}^{(0)}+\hbar^2\widehat{W}^{(1)}}\cdot 1,
\ee
where
\be
\widehat{W}^{(-1)}=-\frac{S^2}{4}t_1,\\
\widehat{W}^{(0)}=\frac{1}{2}\sum_{k,m=0}^\infty (2k+1)(2m+1)t_{2k+1}t_{2m+1}\frac{\p}{\p t_{2k+2m+1}}
+\frac{1}{16}t_1,\\
\widehat{W}^{(1)}=\frac{1}{4}\sum_{k,m=0}^\infty(2k+2m+3)t_{2k+2m+3}\frac{\p^2}{\p t_{2k+1}\p t_{2m+1}},
\ee
and we introduced a new parameter
\be\label{Svar}
S=\hbar N.
\ee
From this representation it follows that after the times rescaling the generalized BGW tau-function has a natural genus expansion
\be
\tau_N(\hbar^{-1}{\bf t},\hbar)=\exp\left(\sum_{g=0}^\infty \hbar^{2g-2}{\mathcal F}_g({\bf t},S)\right).
\ee
From the zeroth equation (\ref{KdVvirN}) it immediately follows that
\be
\sum_{k=0}^\infty (2k+1)\tilde{t}_{2k+1}\frac{\p}{\p t_{2k+1}}{\mathcal F}_g({\bf t},S)=\frac{S^2}{2}\delta_{g,0}-\frac{1}{8}\delta_{g,1},
\ee
where the dilaton shift of the time variables is defined by
\be\label{dilsh}
\tilde{t}_k=t_k-\frac{2}{\hbar}\delta_{k,1}.
\ee
Thus, up to the genus zero and genus one contributions, we can express all ${\mathcal F}_g({\bf t},S)$ in terms of the ``moment variables"
\be
T_k=\frac{t_k}{(2-t_1)^k},
\ee
namely
\be
{\mathcal F}_g({\bf t},S)=\tilde {\mathcal F}_{g}({\bf T},S)+\left(\frac{S^2}{2}\delta_{g,0}-\frac{1}{8}\delta_{g,1}\right)\log\left(1-\frac{t_1}{2}\right).
\ee
Moreover, from (\ref{HomoN}) it follows that ${\mathcal F}_g({\bf T},S)$ are the homogeneous functions of degree $g-1$
\be
\left(\sum_{m=1}^{\infty} m\, T_{2m+1}\frac{\p}{\p T_{2m+1}}-\frac{S}{2}\frac{\p}{\p S}\right)\tilde{\mathcal F}_g({\bf T},S)=(g-1)\, \tilde {\mathcal F}_{g}({\bf T},S).
\ee

Thus, the genus $g$ contribution is given by the sum
\be
 \tilde {\mathcal F}_{g}({\bf T},S)=\sum_{k=0}^\infty (-1)^k S^{2k} \tilde {\mathcal F}_{g}^{(k)}({\bf T}),
\ee
where we introduced the polynomials $ \tilde {\mathcal F}_{g}^{(k)}({\bf T})$ such that
\be
\sum_{m=1}^{\infty} m\, T_{2m+1}\frac{\p}{\p T_{2m+1}} \tilde {\mathcal F}_{g}^{(k)}({\bf T})=(g+k-1)\tilde {\mathcal F}_{g}^{(k)}({\bf T}).
\ee
For $k=0$ they coincide with the free energies for BGW tau-function, given in Appendix \ref{A},
\be
\tilde {\mathcal F}_{g}^{(0)}({\bf T})= {\mathcal F}_{g}({\bf T}).
\ee
Using the recursion (\ref{recurN}) we found expressions for ${\mathcal F}_g^{(k)}({\bf T})$ for all $g+k\leq 20$. For $k>0$ and $g+k\leq 8$ they are given in Appendix \ref{C}.

\subsection{Quantum spectral curve}\label{QSC}

The quantum spectral curve\footnote{For more details on quantum spectral curves see \cite{EynQSC} and references therein.} for the generalized BGW tau-function, as for many other examples of the generating functions related to the KP/Toda hierarchies \cite{ALS,Open1,Aenum}, can be derived from the Sato Grassmannian description. Actually, the principal specialisation of any KP tau-function coincides with the first basis vector of the corresponding point of the Sato Grassmannian. Often, the KS algebra contains an operator, which annihilates this basis vector, and namely this operator describes the quantum spectral curve.

 It follows from (\ref{protoQSC}) that for the generalized BGW tau-function this vector is annihilated by a shifted operator $c_N$:
\be
\left(c_N-\frac{1}{\hbar^2}\right)\Phi_1^{(N)}(\lambda)=0.
\ee
Let us introduce a new variable:
\be
x=\lambda^2.
\ee
Then the corresponding wave function
 \be
\Psi_S(x):=\frac{\hbar}{\sqrt{4\pi\, }x^\frac{1}{4}}\,\,e^{\frac{2\sqrt{x}}{\hbar }}\,\Phi_1^{(S\hbar^{-1})}{(\sqrt{x})}
 \ee
is the modified Baker function
\be
 \Psi_S(x)=I_{S\hbar^{-1}}\left(\frac{2\sqrt{x}}{\hbar}\right).
\ee
It satisfies the modified Bessel equation
\be
\left(\hbar^2 x^2\frac{\p^2}{\p x^2}+\hbar^2 x\frac{\p}{\p x}-x-\frac{S^2}{4}\right) \Psi_S(x)=0.
\ee
which is the quantum spectral curve equation for the generalized BGW model. If we introduce the operators
\be
\hat{x}=x,\,\,\,\,\ \hat{y}=\hbar \frac{\p}{\p x}, 
\ee
then we can rewrite the quantum spectral curve equation as
\be
\left(\hat{x}\hat{y}\hat{x}\hat{y}-\hat{x}-\frac{S^2}{4}\right)\Psi_S(x)=0,
\ee
which in the classical limit reduces to the curve
\be
x^2 y^2 -x -\frac{S^2}{4}=0
\ee
or, equivalentely
\be\label{gensc}
y^2=\frac{1}{x}+\frac{S^2}{4x^2}.
\ee
This curve admits a rational parametrization:
\be\label{ratp}
x=\frac{S^2\, (z-1)}{(z-2)^2},\\
y=\frac{z(z-2)}{2S(z-1)},
\ee
thus, the spectral curve is of genus zero. 

The branch points are the zeros of the differential $dx$,
\be
dx=-\frac{S^2 z}{(z-2)^3}dz,
\ee
which do not coincide with the zeros of the differential $dy$,
\be
dy={\frac {{z}^{2}-2z+2}{ 2\left( z-1 \right) ^{2}S}} dz.
\ee
We see, that on the curve (\ref{ratp}) there is only one branch point, 
\be
z=0,
\ee
which corresponds to
\be
y=0,\,\,\,\,\, x=-\frac{S^2}{4}.
\ee

For the BGW model, that is for $S=0$, the quantum spectral curve equation reduces to
\be
\left(\hat{y}\hat{x}\hat{y}-1\right)\Psi_0(x)=0
\ee
and  in the classical limit 
\be\label{simsc}
y^2=\frac{1}{x}.
\ee
In this limit $y$ plays the role of the global rational coordinate.
This can be considered as the curve for the r-spin intersection numbers with $r=-2$.

We claim that the Chekhov--Eynard--Orantin topological reduction \cite{Eyn1,Eyn2,Eyn3} for the spectral curves (\ref{gensc}) and (\ref{simsc}) should give the expressions for the correlation functions of the generalized BGW and BGW models correspondingly. However, in the next section we will derive the recursion relation for the correlation functions using only the Virasoro constraints (\ref{KdVvirN}).

\subsection{Correlation functions}\label{cor}

The Virasoro constraints can also be reformulated in terms of the correlation functions (multiresolvents). This reformulation leads to the loop equations \cite{Loop1,Loop2,Loop3,Loop4,Loop5,Loop6}.

Sometimes the loop equations can be solved systematically, producing simple recursive relations for the correlation functions \cite{AMM0,AMM1,AMM2,AMM3}. 
Let use define the connected correlation functions 
\be\label{corfunc}
W_{g,n}(x_1,\dots,x_n):=\widehat{\nabla}(x_1)\widehat{\nabla}(x_2)\dots\widehat{\nabla}(x_n) \mathcal{F}_g({\bf t},S)\Big|_{{\bf t}=0},
\ee
where
\be
\widehat{\nabla}(x)=\sum_{k=1}^\infty \frac{1}{x^{k+1}}\frac{\p}{\p t_{2k+1}}.
\ee
Obviously, the correlation functions are symmetric functions of the variables $x_1,\dots,x_n$ and contain all information about the tau-function. 

From the Virasoro constraints (\ref{KdVvirN}) it follows that the correlation functions of the generalized BGW tau-function satisfy the loop equations:
\be\label{Weq}
W_{g,m+1}(x,x_1,\dots,x_m)=\frac{1}{4}W_{g-1,m+2}(x,x,x_1,\dots,x_m)+\left(\frac{1}{16x}\delta_{g,1}-\frac{S^2}{4x}\delta_{g,0}\right)\delta_{m,0}\\
+\frac{1}{4}\sum_{q+p=g,\,I\cup J=\{1,2,\dots,m\}} W_{q,m_1+1}(x,x_{i_1},\dots,x_{i_{m_1}})W_{p,m_2+1}(x,x_{j_1},\dots,x_{j_{m_2}})\\
+\sum_{i=1}^m\left(x_i\frac{\p}{\p x_i}+\frac{1}{2}\right)\frac{W_{g,m}(x_1,\dots,x_{i-1},x,x_{i+1},\dots,x_m)-W_{g,m}(x_1,\dots,x_m)}{x-x_i}
\ee
for all $m\geq 0$ and $g\geq 0$. This is a simple $S$-deformation of the loop equations for the BGW tau-function, which were derived in \cite{IMMM3}.

The simplest case is $g=m=0$, and in this case (\ref{Weq}) gives a quadratic equation for $W_{0,1}$:
\be
W_{0,1}(x)^2-4W_{0,1}(x)-\frac{S^2}{x}=0
\ee
so that
\be
W_{0,1}(x)=2\left(1-\sqrt{1+\frac{S^2}{4x}}\right),
\ee
or
\be
W_{0,1}(x)=2\sum_{k=0}^\infty \left(-\frac{S^2}{8x}\right)^{k+1}\frac{(2k-1)!!}{(k+1)!}.
\ee

This allows us to solve recursively the Loopequations (\ref{Weq}) for $g+m>0$, 
\begin{lemma}
\be\label{lloprec}
W_{g,m+1}(x,x_1,\dots,x_m)\\
=\frac{1}{\sqrt{1+\frac{S^2}{4x}}}\left(\frac{1}{4}\sum'_{q+p=g, I\cup J=\{1,2,\dots,m\}} W_{q,m_1+1}(x,x_{i_1},\dots,x_{i_{m_1}})W_{p,m_2+1}(x,x_{j_1},\dots,x_{j_{m_2}})\right.\\
+\sum_{i=1}^m\left(x_i\frac{\p}{\p x_i}+\frac{1}{2}\right)\frac{W_{g,m}(x_1,\dots,x_{i-1},x,x_{i+1},\dots,x_m)-W_{g,m}(x_1,\dots,x_m)}{x-x_i}\\
\left.+\frac{1}{4}W_{g-1,m+2}(x,x,x_1,\dots,x_m)+\frac{1}{16x}\delta_{g,1}\delta_{m,0}\right),
\ee
where we exclude from the sum two terms: with $q=g$, $I=\{1,2,\dots,m\}$, $J=\{\emptyset\}$ and with $p=g$, $I=\{\emptyset\}$, $J=\{1,2,\dots,m\}$.
\end{lemma}

In particular, the genus zero two-point function is
\be
W_{0,2}(x_1,x_2)=\frac{1}{\sqrt{1+\frac{S^2}{4x_1}}}\left(x_2\frac{\p}{\p x_2}+\frac{1}{2}\right)\frac{W_{0,1}(x_1)-W_{0,1}(x_2)}{x_1-x_2}\\
=\frac{1}{2(x_1-x_2)^2}\left(\frac{S^2+2(x_1+x_2)}{\sqrt{1+\frac{S^2}{4x_1}}\sqrt{1+\frac{S^2}{4x_2}}}-2(x_1+x_2)\right).
\ee
It is regular at the coincident points (that is when $x_1=x_2$ and $y_1=y_2$), and has a second order pole when the points are on different sheets above the same base point (that is when $x_1=x_2$ and $y_1=-y_2$),
\be
W_{0,2}(x_1,x_2)=-\frac{4x_1}{(x_1-x_2)^2}+\dots.
\ee
 
In genus one we have
\be
W_{1,1}(x)=\frac{1}{\sqrt{1+\frac{S^2}{4x}}}\left(\frac{1}{4}W_{0,2}(x,x)+\frac{1}{16x}\right)=\frac{1}{2^4 x\left(1+\frac{S^2}{4x}\right)^\frac{5}{2}},
\ee
or
\be
W_{1,1}(x)=\frac{1}{2^4\cdot 3 x}\sum_{k=0}^\infty \left(-\frac{S^2}{8x}\right)^{k}\frac{(2k+3)!!}{k!}.
\ee
On the next level of recursion we have
\be
W_{0,3}(x_1,x_2,x_3)=-\frac{S^2}{8x_1x_2x_3\sqrt{1+\frac{S^2}{4x_1}}\sqrt{1+\frac{S^2}{4x_2}}\sqrt{1+\frac{S^2}{4x_3}}},
\ee
\be
W_{1,2}(x_1,x_2)=\frac{{S}^{8}-6 \left( \,x_{{1}}+\,x_{{2}} \right) {S}^{6}-136\,{S}^{4}x_{
{2}}x_{{1}}-128\,x_{{1}}x_{{2}} \left( x_{{2}}+x_{{1}} \right) {S}^{2}
+128\,{x_{{1}}}^{2}{x_{{2}}}^{2}
}{2^{12}x_1^3 x_2^3 \left(1+\frac{S^2}{4x_1}\right)^\frac{7}{2}\left(1+\frac{S^2}{4x_2}\right)^\frac{7}{2}},
\ee
\be
W_{2,1}(x)=\frac{S^4-20S^2x+9x^2}{2^{10}x^4\left(1+\frac{S^2}{4x}\right)^\frac{11}{2}}.
\ee

For the stable cases (the cases with $2g+n-2>0$, that is, for all $W_{g,n}$'s except for $W_{0,1}$ and $W_{0,2}$) let use define the differentials forms
\be
\omega_{g,n}(z_1,\dots,z_n):=S^{2g-2+n}W_{g,n}(x_1,\dots,x_n)d\sqrt{x_1}\dots d\sqrt{x_n}\\
=S^{2g-2+n}\frac{W_{g,n}(x_1,\dots,x_n)}{2^n\sqrt{x_1\dots x_n}}d x_1\dots dx_n.
\ee
They satisfy the recursion relations which follow from (\ref{lloprec}) and can be easily found for small $g$ and $n$. In particular,
\be
\omega_{1,1}(z)={\frac {z-1}{z^{4}}} dz
\ee
\be
\omega_{2,1}(z)={\frac { \left( 105-210\,z+133\,{z}^{2}-28\,{z}^{3}+{z}^{4} \right) 
 \left( z-2 \right) ^{2} \left( z-1 \right) }{{z}^{10}}}
dz
\ee
\be
\omega_{1,2}(z_1,z_2)=\left(54\,{z_{{2}}}^{2}{z_{{1}}}^{4}+24\,{z_{{2}}}^{3}{z_{{1}}}^{3}-14\,{z_{
{2}}}^{3}{z_{{1}}}^{4}+54\,{z_{{2}}}^{4}{z_{{1}}}^{2}-14\,{z_{{2}}}^{4
}{z_{{1}}}^{3}\right.\\
+{z_{{2}}}^{4}{z_{{1}}}^{4}+24\,{z_{{1}}}^{2}{z_{{2}}}^{
2}-80\,{z_{{1}}}^{4}z_{{2}}-24\,{z_{{1}}}^{3}{z_{{2}}}^{2}-24\,{z_{{1}
}}^{2}{z_{{2}}}^{3}\\
\left.-80\,z_{{1}}{z_{{2}}}^{4}+40\,{z_{{2}}}^{4}+40\,{z_
{{1}}}^{4}
\right){\frac {1 }{{z_{{1}}}^{6}{z_{{2}}}^{6}}}
 dz_1 dz_2
\ee
\be
\omega_{0,3}(z_1,z_2,z_3)=\frac{8}{z_1^2z_2^2z_3^2}dz_1dz_2dz_3
\ee

\begin{conjecture}
All  $\omega_{g,n}$ are the meromorphic differentials, defined on the spectral curve (\ref{ratp}) and symmetric in $z_j$'s. Moreover, for any $g$ and $n$
\be
\frac{z_1^2\dots z_n^2\,\omega_{g,n}(z_1,\dots,z_n)}{d z_1\dots dz_n}
\ee
is a polynomial in each of variables $z_1^{-1},\dots,z_n^{-1}$. Thus, $\omega_{g,n}$ have poles of finite degree only at the branch point $z_j=0$.
\end{conjecture}

It should be simple to prove this conjecture using the Chekhov--Eynard--Orantin topological recursion methods, which are beyond the scope of this paper.

In the limit of $S=0$ the correlation functions $W_{g,n}$ coincide with the original BGW model. In this case all $W_{g,n}$ are polynomials in $x_j^{-1}$, see \cite{IMMM3}.

\subsection{Genus zero contribution}

Formulas (\ref{F01})-(\ref{F09}) as well as higher terms indicate that the coefficients of the expansion of genus zero free energy are quite simple.
\begin{conjecture}
\begin{multline}
\left[T_3^{j_1}T_5^{j_2}T_7^{j_3}\dots\right] \tilde{\mathcal F}_0({\bf T},S)\\
=\frac{(-1)^{m+1}\,(3j_1+5j_3+7j_3+\dots-1)!}{2^m\,(2m+2)!}\frac{(3!!)^{j_1}\,(5!!)^{j_2}\,(7!!)^{j_3}\dots}{(1!)^{j_1}\,j_1!\,(2!)^{j_2} \,j_2! \,(3!)^{j_3}\,j_3!\,\dots}\,S^{2m+2},
\end{multline}
where
\be
m=j_1+2j_2+3j_3+\dots.
\ee
\end{conjecture}
From this conjecture and from the definition of the variables $T_k$ it immediately follows that
\be
{\mathcal F}_0({\bf t},S)=\sum_{\substack{j_0,j_1,j_2,\dots\\ j_0+j_1+\dots>0}}A(j_0,j_1,j_2,\dots)\,\frac{(-1)^{m+1}S^{2m+2}\,(j_0+3j_1+5j_2+\dots-1)!}{2^{m+j_0+3j_1+5j_2+\dots}\,(2m+2)!}\,t_1^{j_0}t_3^{j_1}t_5^{j_2}\dots,
\ee
where
\be
A(j_0,j_1,j_2,\dots)=\frac{(1!!)^{j_0}\,(3!!)^{j_1}\,(5!!)^{j_2}\dots}{(0!)^{j_0}\,j_0!\,(1!)^{j_1}\,j_1!\,(2!)^{j_2}\,j_2!\dots}.
\ee
This expression is consistent with expressions for the correlation functions obtained in Section \ref{cor}.
It should help to identify the coefficients of the generalized BGW model with the enumerative geometry invariants. This conjecture probably can be proved with the help of the Baker--Campbell--Hausdorff analysis of (\ref{BCH}).

Let us give for comparison an expression for the genus zero free energy of the Kontsevich--Witten tau-function
\be
{\mathcal F}_0^{KW}({\bf t},S)=\sum_{j_0,j_1,j_2,\dots}A(j_0,j_1,j_2,\dots)\,m!\,\delta(-j_0+j_2+2j_3+\dots+3)\,t_1^{j_0}t_3^{j_1}t_5^{j_2}\dots.
\ee

\subsection{Polynomial tau-functions of KdV hierarchy}\label{Polysec}

It appears that for $N-\frac{1}{2} \in \mathbb Z$ the generalized BGW tau-function is a polynomial in times. From (\ref{Nsym}) it follows that it is enough to consider only positive values of $N$. In this section we assume that 
\be
l=N-\frac{1}{2}\in {\mathbb N}_0.
\ee
Then we have
\begin{theorem}
For the half-integer value of $N$ the generalized BGW tau-function is polynomial in times. Moreover, up to the dilaton shift of the times, it is equal to the the Schur function corresponding to the triangular partition of $\frac{l(l+1)}{2}$,
\be
\lambda(l)=\left(l,l-1,l-2,\dots,1\right).
\ee
Namely
\be\label{prop}
\tau_{l+\frac{1}{2}}({\bf t})=C_l\, s_{\lambda(l)}(\tilde{\bf t})
\ee
where the dilaton shift is given by (\ref{dilsh}) and
\be
C_l = \frac{(-\hbar)^{\frac{l(l+1)}{2}}}{2^{l^2}}\prod_{k=1}^l\frac{(2l-2k+1)!}{(l-k)!}.
\ee
\end{theorem}

\begin{proof}

In this case all sums in the expressions for the basic vectors (\ref{phiser}) have only a finite numbers of terms:
\be
\Phi^{(l+\frac{1}{2})}_j(\lambda)=
\begin{cases}
\displaystyle{\lambda^{j-1}+\sum_{k=1}^{l-j+1}(-\hbar)^k\frac{a_k(j-l-\frac{1}{2})}{16^k\, k!}\,\lambda^{j-k-1} \,\,\,\,\,\,\,\,\,\,\,\,\,\,\,\,\,\,\, \mathrm{for} \quad j\leq l+1},\\[10pt]
\displaystyle{ \lambda^{j-1}+\sum_{k=1}^{j-l-2}(-\hbar)^k\frac{a_k(j-l-\frac{1}{2})}{16^k\, k!}\,\lambda^{j-k-1}\,\,\,\,\,\,\,\,\,\,\,\,\,\,\,\,\,\,\, \mathrm{for} \quad j>l+1}.
\end{cases}
\ee
Thus, $\Phi_j^{(l+1/2)}$ is not polynomial (contains negative powers of $\lambda$) only for $j<\frac{l}{2}+1$,
and in this case the most singular term is proportional to $\lambda^{2j-l-2}$. 
Moreover,
\be\label{boll}
\Phi^{(l+\frac{1}{2})}_{l+1}(\lambda)=\lambda^{l},\\
\Phi^{(l+\frac{1}{2})}_{l+2}(\lambda)=\lambda^{l+1},
\ee
thus, from (\ref{KS}) we see that 
\be\label{boll1}
\lambda^{k} \in {\mathcal W}_{l+\frac{1}{2}} \quad \mathrm{for}\quad k\geq l.
\ee

Therefore, for any $M \geq l$ a ratio of determinants
\be
\tau_{l+\frac{1}{2}}(\Lambda)=\frac{\det_{i,j=1}^M{\left(\Phi^{(l+\frac{1}{2})}_{j}(\lambda_i)\right)}}{\Delta(\lambda)}
\ee
is a symmetric polynomial (not homogeneous!) in the eigenvalues $\lambda_j^{-1}$ of total degree $\frac{l(l+1)}{2}$. It means that if we put $\deg t_k=k$, $\tau_{l+\frac{1}{2}}$ is a polynomial in times $t_k$ of degree  $\frac{l(l+1)}{2}$, for example
\be
\tau_{\frac{1}{2}}=1,\\
\tau_{\frac{3}{2}}=1-\hbar \frac{t_1}{2}=-\frac{\hbar}{2}\tilde{t_1},\\
\tau_{\frac{5}{2}}=1-\frac{3}{2}\hbar\,t_{{1}}+\frac{3}{4}\,\hbar^2\,{t_{{1}}}^{2}+\frac{3}{8}\hbar^3\,t_{{3}}-\frac{1}{8}\hbar^3\,{t_{{1}}}^{3}=\hbar^3\left(\frac{3}{8}\,\tilde{t}_{{3}}-\frac{1}{8}\,{\tilde{t}_{{1}}}^{3}\right).
\ee

 Let us prove, that the tau-unctions $\tau_{l+\frac{1}{2}}$ actually coincide (up to a constant normalization) with the Schur functions. The shift of the times in the tau-functions corresponds to the action of the multiplication operator on the Sato Grassmannian. Namely, if a given tau-function $\tau({\bf t})$ is described by the point $\mathcal{W}$ of the Sato Grassmannian, the tau-function
\be
\tilde{\tau}(\bf t)=\tau({\bf t+a})
\ee
corresponds to the point of the Sato Grassmannian, specified by
\be\label{tildbv}
\widetilde{\mathcal W} =e^{\sum \frac{a_k \lambda^k}{k}}{\mathcal W}.
\ee

In particular, to get rid of the dilaton shift (\ref{dilsh}) we introduce 
\be
\tilde{\tau}_{l+\frac{1}{2}}({\bf t})={\tau}_{l+\frac{1}{2}}\left({\bf t}+\frac{2}{\hbar}\delta_{k,1}\right),
\ee
which corresponds to the point of the point of Sato Grassmannian
\be
\widetilde{\mathcal W}_{l+\frac{1}{2}}=e^\frac{2\lambda}{\hbar}\,{\mathcal W}_{l+\frac{1}{2}}.
\ee
Let us show that
\be\label{tildsg}
\lambda^{2k-2-l}\in \widetilde{\mathcal W}_{l+\frac{1}{2}} \quad \mathrm{for}\quad l\geq k>0.
\ee
Indeed, from (\ref{tildbv}) it follows that
\be\label{tildbild}
e^\frac{2\lambda}{\hbar} \sum_{j=1}^\infty \alpha_j \Phi^{(l+\frac{1}{2})}_j(\lambda)\,\in\, \widetilde{\mathcal W}_{l+\frac{1}{2}}
\ee
for any constants $\alpha_j$. In particular, from (\ref{Phi0}) and (\ref{PhiN}) it follows that if we choose these constants such that $\sum_{j=1}^\infty \alpha_j \hbar^j t^{1-j} =\exp(t^{-1})$, then
(\ref{tildbild}) reduces to 
\be
\lambda^{l+1}\int_\gamma e^{\tilde\lambda^2 t} \frac{dt}{t^{-l+\frac{1}{2}}}\in\, \widetilde{\mathcal W}_{l+\frac{1}{2}}.
\ee 
Since
\be
\int_\gamma e^{\tilde\lambda^2 t} \frac{dt}{t^{-l+\frac{1}{2}}}\sim \lambda^{-2l-1}
\ee
it proves (\ref{tildsg}) for $k=1$. To prove (\ref{tildsg}) for $l\geq k>1$ one have just to choose other values for the constants $\alpha_j$.

Thus
\be
\widetilde{\mathcal W}_{l+\frac{1}{2}}=\left<\lambda^{-l},\lambda^{2-l},\lambda^{4-l},\dots,\lambda^{l-4},\lambda^{l-2},\lambda^{l},\lambda^{l+1},\dots\right>
\ee
It does not belong to the big cell of the Sato Grassmannian, and the KdV tau-function is (up to  a constant factor) the Schur function, corresponding to the triangular Young tableau with
\be\label{lambdal}
\lambda(l)=\left(l,l-1,l-2,\dots,1\right).
\ee
These KdV tau-functions were described already in \cite{JimbM}. 
The constant $C_l$ in (\ref{prop}) can be easily found from the comparison of the r.h.s. and the l.h.s. for ${\bf t}=0$. Namely,
\be
s_\lambda({\bf t})\Big|_{t_k=0,k>1}= s_{\lambda}(t_k=\delta_{k,1})\, t_1^{|\lambda|},
\ee
thus
\be
C_l = \frac{\hbar^ \frac{l(l+1)}{2}}{(-2)^{\frac{l(l+1)}{2}}s_{\lambda(l)}(t_k=\delta_{k,1})}=\frac{(-\hbar)^{\frac{l(l+1)}{2}}}{2^{l^2}}\prod_{k=1}^l\frac{(2l-2k+1)!}{(l-k)!}.
\ee
\end{proof}
We have a corollary 
\begin{lemma}
The tau-function of the KdV hierarchy given by the Schur function
\be
\tau({\bf t})=s_{\lambda(l)}({\bf t}),
\ee
where the partition is given by (\ref{lambdal}),
is uniquely (up to normalization) specified by the Virasoro constraints
\be
 \widehat{\mathcal L}_m\, \tau({\bf t})=0, \,\,\,\,\, m\geq 0,
\ee
where
\be
\widehat{\mathcal L}_m= \frac{1}{2}\sum_{k=0}^\infty (2k+1) \tilde{t}_{2k+1} \frac{\p}{\p t_{2k+2m+1}}+\frac{1}{4} \sum_{a+b=m-1} \frac{\p^2}{\p t_{2a+1} \p t_{2b+1}}-\frac{l(l+1)}{4}\delta_{m,0}.
\ee
\end{lemma}

Polynomially of the tau-function (\ref{prop}) means that in this case the expansion (\ref{sumhbar}) is finite, and
\be
\left(\widehat{W}_{l+\frac{1}{2}}\right)^{\frac{l(l+1)}{2}+1} \cdot 1=0.
\ee
Also, since
\be
B_k\left(l+\frac{1}{2}\right)=0,\,\,\,\,\,k>|l|,
\ee
the terms $\tau^{(k)}_N$, which are polynomials in $N$, are indeed divisive by $B_l(N)$ for $k\leq\frac{l(l+1)}{2}$.

Thus, it is natural to express the free energy in terms of only $B_m(N)$ and $T_k$.
Namely, for 
\be
\tau_N(\hbar^{-1}{\bf t},\hbar)=\exp\left({\mathcal F}({\bf t},N,\hbar)\right)
\ee
where we do not introduce the variable $S$, (\ref{Svar}), we have
\be
\mathcal{F}({\bf t},N,\hbar)=\frac{4N^2-1}{8}\log\left(1-\frac{t_1}{2}\right)+\sum_{k=2}\hbar^{2g-2}
F_g({\bf T},N)
\ee
where
\be
F_g({\bf T},N)=\sum_{k=2}^g B_k(N) F_{g,k}(\bf T).
\ee
Polynomials $F_{g,k}$ can be found using the cut-and-join operator or from (\ref{prop}). In Appendix \ref{D} we give expressions for $g\leq 6$.

\section{Concluding remarks}

In this paper we have investigated the generalized BGW tau-function.
Obtained results are not only interesting for the matrix model theory, but also should help to identify the generalized BGW tau-function with a generating function of some enumerative geometry invariants. A natural candidate would be a version of open $r=-2$ spin intersection numbers. However, probably this interpretation is too naive. One of the reasons is that the introduction of the variable $N$, which should add boundaries to the theory, is not accompanied by new variables for the descendants on the boundary (which appears in the Kontsevich-Penner model and constitute a second infinite set of times of the KP tau-function). 

The results also should help  to develop the theory of the Givental decomposition. The cut-and-join representation (in particular, its free field version) should allow us to represent the decomposition formulas purely in terms of simple exponential operators. Of course, the same analysis can be applied to other antipolynomial generalized Kontsevich models. 

This paper contains all necessary prerequisites for construction of the Chekhov--Eynard--Orantin topological recursion for the generalized BGW model, namely, the quantum and classical spectral curves, rational parametrization, wave function, one and two point correlation functions in genus zero and loop equations. 
It would be interesting to compere our results with the contour integral expressions for the $n$-point (all-genera) correlation functions obtained in \cite{BrezinHik} and with the recursion relations for the KdV hierarchy correlation functions from \cite{Dub}. It is also interesting to find the compact expressions for the higher genera contributions to the free energy. Some compact expressions for the higher genera contributions to the free energy in terms of the moments are given in \cite{ChekhovKW}, but their conclusions about the relation of this model with the Kontsevich-Witten tau-function and the structure of the Virasoro constraints look to be not completely consistent with our results. These topics are beyond the scope of the present paper and will be considered later.

\section*{Acknowledgments}
The author is grateful to L.~Chekhov and A.~Mironov for useful discussions. This paper is a first part of the project, initiated during   the OIST workshop on Moduli Space, Conformal Field Theory and Matrix Models 2015, and the author is grateful to the organizers of this workshop. This work was supported in part by IBS-R003-D1, by the Natural Sciences and Engineering Research Council of Canada (NSERC), by the Fonds de recherche du Qu\'ebec Nature et technologies (FRQNT) and by RFBR grants 15-01-04217 and 15-52-50041YaF. 

\newpage
\begin{appendices}\label{appp}
%\newpage

%\appendix
\section{Free energy of BGW tau-function}
\label{A}
%\addcontentsline{toc}{section}{Appendix A}
\def\theequation{A\arabic{equation}}
\setcounter{equation}{0}

\be
{\mathcal F}_{2}={\frac {9}{128}}\,T_{{3}}
\ee
\be
{\mathcal F}_{3}={\frac {567}{1024}}\,{T_{{3}}}^{2}+{\frac {225}{1024}}\,T_{{5}}
\ee
\be
{\mathcal F}_4={\frac {64989}{4096}}\,{T_{{3}}}^{3}+{\frac {388125}{32768}}\,T_{{5}}T
_{{3}}+{\frac {55125}{32768}}\,T_{{7}}
\ee
\be
{\mathcal F}_5={\frac {70864875}{65536}}\,T_{{5}}{T_{{3}}}^{2}+{\frac {14123025}{
65536}}\,T_{{7}}T_{{3}}+{\frac {6251175}{262144}}\,T_{{9}}+{\frac {
130301217}{131072}}\,{T_{{3}}}^{4}+{\frac {28252125}{262144}}\,{T_{{5}
}}^{2}
\ee
\be
{\mathcal F}_6={\frac {286765250859}{2621440}}\,{T_{{3}}}^{5}+{\frac {2269176525}{
4194304}}\,T_{{11}}+{\frac {37656646875}{1048576}}\,T_{{3}}{T_{{5}}}^{
2}\\+{\frac {18826455375}{524288}}\,T_{{7}}{T_{{3}}}^{2}+{\frac {
25035955875}{4194304}}\,T_{{3}}T_{{9}}+{\frac {12519714375}{2097152}}
\,T_{{7}}T_{{5}}+{\frac {81770259375}{524288}}\,T_{{5}}{T_{{3}}}^{3}
\ee
\be
{\mathcal F}_7={\frac {27537582342375}{16777216}}\,T_{{9}}{T_{{3}}}^{2}+{\frac {
7852650127875}{33554432}}\,T_{{9}}T_{{5}}+{\frac {206914899886875}{
16777216}}\,{T_{{3}}}^{2}{T_{{5}}}^{2}\\
+{\frac {34484117212125}{4194304
}}\,T_{{7}}{T_{{3}}}^{3}+{\frac {34539827452875}{1048576}}\,T_{{5}}{T_
{{3}}}^{4}+{\frac {9180336943125}{16777216}}\,{T_{{5}}}^{3}\\
+{\frac {
1963178035875}{16777216}}\,{T_{{7}}}^{2}+{\frac {19600404065991}{
1048576}}\,{T_{{3}}}^{6}+{\frac {602628451425}{33554432}}\,T_{{13}}\\
+{
\frac {13769702800875}{4194304}}\,T_{{7}}T_{{5}}T_{{3}}+{\frac {
3925999556325}{16777216}}\,T_{{11}}T_{{3}}

\ee
\be
{\mathcal F}_8={\frac {169591162989488625}{67108864}}\,T_{{7}}{T_{{3}}}^{4}+{\frac {
26488676216338875}{2147483648}}\,T_{{11}}T_{{5}}\\
+{\frac {
26487328325483025}{2147483648}}\,T_{{13}}T_{{3}}+{\frac {
26488802802632625}{2147483648}}\,T_{{9}}T_{{7}}+{\frac {
75275622313203375}{134217728}}\,T_{{9}}{T_{{3}}}^{3}\\
+{\frac {
6632707287504375}{67108864}}\,T_{{3}}{T_{{7}}}^{2}+{\frac {
339191251470703125}{67108864}}\,{T_{{5}}}^{2}{T_{{3}}}^{3}+{\frac {
645210015875843625}{67108864}}\,T_{{5}}{T_{{3}}}^{5}\\
+{\frac {
13265904719221875}{134217728}}\,T_{{7}}{T_{{5}}}^{2}+{\frac {
13264818560895825}{134217728}}\,T_{{11}}{T_{{3}}}^{2}+{\frac {
75280511033859375}{134217728}}\,{T_{{5}}}^{3}T_{{3}}\\
+{\frac {
538246474955839917}{117440512}}\,{T_{{3}}}^{7}+{\frac {
1762688220418125}{2147483648}}\,T_{{15}}+{\frac {26530744498689375}{
134217728}}\,T_{{9}}T_{{5}}T_{{3}}\\
+{\frac {112917280552295625}{
67108864}}\,T_{{5}}{T_{{3}}}^{2}T_{{7}}
\ee
\be
{\mathcal F}_9=
{\frac {77431411127351806875}{1073741824}}\,{T_{{7}}}^{2}{T_{{3}}}^{2}
+{\frac {258385134479852784375}{536870912}}\,{T_{{5}}}^{3}{T_{{3}}}^{2
}\\
+{\frac {2715656473902641360625}{1073741824}}\,{T_{{5}}}^{2}{T_{{3}}}
^{4}+{\frac {258375760190755743375}{1073741824}}\,T_{{9}}{T_{{3}}}^{4}\\
+{\frac {32565363218292436875}{4294967296}}\,T_{{9}}{T_{{5}}}^{2}+{
\frac {1806818862379174875}{2147483648}}\,T_{{11}}T_{{7}}\\
+{\frac {
51619703670742474575}{1073741824}}\,T_{{11}}{T_{{3}}}^{3}+{\frac {
3613628364405184125}{4294967296}}\,T_{{13}}T_{{5}}\\
+{\frac {
8141354220180706875}{1073741824}}\,{T_{{7}}}^{2}T_{{5}}+{\frac {
16281810076944587175}{2147483648}}\,T_{{13}}{T_{{3}}}^{2}\\
+{\frac {
1806755425928578125}{2147483648}}\,T_{{15}}T_{{3}}+{\frac {
849028159501396875}{17179869184}}\,T_{{17}}\\
+{\frac {
309725220277322836875}{2147483648}}\,T_{{5}}{T_{{3}}}^{2}T_{{9}}+{
\frac {16282283116759356375}{1073741824}}\,T_{{9}}T_{{3}}T_{{7}}\\
+{
\frac {16282250837254450125}{1073741824}}\,T_{{11}}T_{{3}}T_{{5}}+{
\frac {154866013725681129375}{1073741824}}\,{T_{{5}}}^{2}T_{{3}}T_{{7}
}\\
+{\frac {129190312318002901875}{134217728}}\,T_{{5}}{T_{{3}}}^{3}T_{{
7}}+{\frac {103246107731883140625}{8589934592}}\,{T_{{5}}}^{4}\\
+{\frac 
{7227277861688852625}{17179869184}}\,{T_{{9}}}^{2}+{\frac {
996625993639547388375}{268435456}}\,T_{{5}}{T_{{3}}}^{6}\\
+{\frac {
271561567871335604625}{268435456}}\,T_{{7}}{T_{{3}}}^{5}
\ee

%\appendix
\section{Coefficients of generalized BGW tau-function}
\label{B}
%\addcontentsline{toc}{section}{Appendix B}
\def\theequation{B\arabic{equation}}
\setcounter{equation}{0}

\be
\tau^{(1)}_N=\frac{B_1(N)}{2^4\, 1!}\,  t_{{1}}
\ee

\be
\tau^{(2)}_N =\frac{B_2(N)}{2^7\,2!}  t_1^2
\ee

\be
\tau^{(3)}_N ={\frac {B_2(N)}{2^{11}\,3!}}\,   \left(24\,t_{{3}} - 4\,{t_{{1}}}^{3}
{N}^{2}+17\,{t_{{1}}}^{3}\right) 
\ee

\be
\tau^{(4)}_N ={\frac {B_3(N)}{2^{16}\,4!}}\, \left(96\,t_{{3}}- 4\,{t_{{1}}}^{3}{N}^{2}+17\,{t_{{1}}}^{3} \right) 
t_{{1}} 
\ee

\be
\tau^{(5)}_N ={\frac { B_3(N)}{2^{20}\,5!}}\,  \left( 16\,{t_{{1}}}^{5}{N}^{4}-200\,
{t_{{1}}}^{5}{N}^{2}-960\,{t_{{1}}}^{2}t_{{3}}{N}^{2}+3840\,t_{{5}}+
7920\,{t_{{1}}}^{2}t_{{3}}+561\,{t_{{1}}}^{5} \right) 
\ee

\be
\tau^{(6)}_N ={\frac {B_3(N)}{2^{24}\,6!}}\, \left(7680\,{t_{{1}}}^{3}t_{{3}}{N}^{4} -64\,{N}
^{6}{t_{{1}}}^{6}+1456\,{t_{{1}}}^{6
}{N}^{4}-142080\,{t_{{1}}}^{3}t_{{3}}{N}^{2}+23001\,{t_{{1}}}^{6}\right.\\
\left.-92160\,t_{{1}}t_{{5}}{N}^
{2}-10444\,{t_{{1}}}^{6}{N}^{2}-23040\,{t_{{3}}}^{2}{N}^{2}+649440\,{t
_{{1}}}^{3}t_{{3}}+944640\,t_{{1}}t_{{5}}+466560
\,{t_{{3}}}^{2} \right)  
\ee

\be
\tau^{(7)}_N ={\frac {B_4(N)}{2^{28}\,7!}}\,  \left( 1612800\,t_{{7}}-10444\,{t_{{1}}}^{7}{N}^{2}-64\,{N}^{6}{t_{{1}}}^{7}+23001\,{t_{{1}}}
^{7}\right.\\
+1456\,{t_{{1}}}^{7}{N}^{4}+13440\,{t_{{1}}}^{4}t_{{3}}{N}^{4}-
248640\,{t_{{1}}}^{4}t_{{3}}{N}^{2}+1136520\,{t_{{1}}}^{4}t_{{3}}-
322560\,{t_{{1}}}^{2}t_{{5}}{N}^{2}\\
\left.+3306240\,{t_{{1}}}^{2}t_{{5}}+
3265920\,t_{{1}}{t_{{3}}}^{2}-161280\,t_{{1}}{t_{{3}}}^{2}{N}^{2}
 \right)  
\ee

\be
\tau^{(8)}_N ={\frac {B_4(N)}{2^{32}\,8!}}\, \left(1311057\,{t_{{1}}}^{8} -687312\,{t_{{1}}}^{8}{N}^{2}+124768\,{t_{{1}}}
^{8}{N}^{4}+256\,{N}^{8}{t_{{1}}}^{8}\right.\\
-86016
\,{t_{{1}}}^{5}{N}^{6}t_{{3}}+2817024\,{t_{{1}}}^{5}t_{{3}}{N}^{4}-
29949696\,{t_{{1}}}^{5}t_{{3}}{N}^{2}+103650624\,{t_{{1}}}^{5}t_{{3}}\\
+
502548480\,{t_{{1}}}^{3}t_{{5}}+3440640\,{t_{{1}}}^{3}t_{{5}}{N}^{4}-
84295680\,{t_{{1}}}^{3}t_{{5}}{N}^{2}+2580480\,{t_{{1}}}^{2}{t_{{3}}}^
{2}{N}^{4}\\
+744629760\,{t_{{1}}}^{2}{t_{{3}}}^{2}-89026560\,{t_{{1}}}^{
2}{t_{{3}}}^{2}{N}^{2}-51609600\,t_{{1}}t_{{7}}{N}^{2}\\
\left.+735436800\,t_{{
1}}t_{{7}}+727695360\,t_{{5}}t_{{3}}-20643840\,t_{{5}}t_{{3}}{N}^{2}
-9472\,{N}^{6}{t_{{1}}}^{8}
\right)
\ee

\be
\tau^{(9)}_N ={\frac {B_4(N)}{2^{36}\,9!}}\,\left(105345515520\,t_{{9}}-28333670400\,{t_{{1}}}^{2}t_{{7}}{N}^{2}+
928972800\,{t_{{1}}}^{2}t_{{7}}{N}^{4}\right.\\
-26295736320\,{t_{{1}}}^{3}{t_{{
3}}}^{2}{N}^{2}+1571512320\,{t_{{1}}}^{3}{t_{{3}}}^{2}{N}^{4}-30965760
\,{t_{{1}}}^{3}{t_{{3}}}^{2}{N}^{6}-1024\,{N}^{10}{t_{{1}}}^{9}\\
+61931520\,{t_
{{3}}}^{3}{N}^{4}+1261854720\,{t_{{1}}}^{4}t_{{5}}{N}^{4}-30965760\,{t_{{1}}}^{4}
t_{{5}}{N}^{6}+743178240\,t_{{1}}t_{{5}}t_{
{3}}{N}^{4}\\
+454358016\,{t_{
{1}}}^{6}t_{{3}}{N}^{4}-25288704\,{t_{{1}}}^{6}{N}^{6}t_{{3}}+516096\,
{N}^{8}{t_{{1}}}^{6}t_{{3}}+425701785600\,t_{{1}}t_{{5}}t_{{3}}\\
-38273679360\,t_{{1}}t_{{5}}t_{{3}}{N}^{2}-3541999104\,{t_{{1}}}^{6}t_{{3}}{N}^{2}+215115264000\,{t_{{1}}}^{2}t_{{7}}\\
+10859168\,{t_{{1}}}^{9}{N}^{4}-1114752\,{N}^{6}{t_{{1}}}^{9
}+10105935840\,{t_{{1}}}^{6}t_{{3}}+73497715200\,{t_{{1}}}^
{4}t_{{5}}\\
-49919508\,{t_{{1}}}^{9}{N}^{2}+145202803200\,{t_
{{1}}}^{3}{t_{{3}}}^{2}-3622993920\,{t_{{3}}}^{3}{N}^{2}+54528\,{N}^{8}{t_{{1}}}^{9}\\
\left.-5202247680\,t_{{9}}{N}^{2}+85218705\,{t_{{1}}}^{9}+
70265180160\,{t_{{3}}}^{3}-16851179520\,{t_{{1}}}^{4}t_{{5}}{
N}^{2}
\right)
\ee

\be
\tau^{(10)}_N ={\frac {B_4(N)}{2^{40}\,10!}}\, \left(1036733644800\,{t_
{{1}}}^{2}t_{{5}}t_{{3}}{N}^{4}-22483928678400\,{t_{{1}}}^{2}t_{{5}}t_{{3}}{N}^{2}\right.\\
-14863564800\,{t_{{1}}}^{2}t_{{5}}t_{{3
}}{N}^{6}-9762729984000\,{t_{{1}}}^{3}t_{{7}}{N}^{2}+603832320000\,{t_
{{1}}}^{3}t_{{7}}{N}^{4}\\
-12386304000\,{t_{{1}}}^{3}t_{{7}}{N}^{6}-
6250999910400\,{t_{{1}}}^{4}{t_{{3}}}^{2}{N}^{2}+549758361600\,{t_{{1}
}}^{4}{t_{{3}}}^{2}{N}^{4}\\
-21366374400\,{t_{{1}}}^{4}{t_{{3}}}^{2}{N}^
{6}+309657600\,{t_{{1}}}^{4}{N}^{8}{t_{{3}}}^{2}-3048253931520\,{t_{{1
}}}^{5}t_{{5}}{N}^{2}\\
+319040225280\,{t_{{1}}}^{5}t_{{5}}{N}^{4}-
14615838720\,{t_{{1}}}^{5}t_{{5}}{N}^{6}+247726080\,{t_{{1}}}^{5}{N}^{
8}t_{{5}}\\
-427128111360\,{t_{{1}}}^{7}t_{{3}}{N}^{2}+67623045120\,{t_{{
1}}}^{7}t_{{3}}{N}^{4}-5233582080\,{t_{{1}}}^{7}{N}^{6}t_{{3}}\\
+198328320\,{N}^{8}{t_{{1}}}^{7}t_{{3}}-2949120\,{N}^{10}{t_{{1}}}^{7}t
_{{3}}+155381151744000\,{t_{{1}}}^{2}t_{{5}}t_{{3}}+4096\,{N}^{12}{t_{{1}}}^{10}\\
-5455392768000\,t_{
{1}}{t_{{3}}}^{3}{N}^{2}+190129766400\,t_{{1}}{t_{{3}}}^{3}{N}^{4}-
2477260800\,t_{{1}}{t_{{3}}}^{3}{N}^{6}\\
-8011461427200\,t_{{1}}t_{{9}}{
N}^{2}+208089907200\,t_{{1}}t_{{9}}{N}^{4}+76202709811200\,t_{{7}}t_{{
3}}\\
-5179952332800\,t_{{7}}t_{{3}}{N}^{2}+74317824000\,t_{{7}}t_{{3}}{N
}^{4}+38097174528000\,{t_{{5}}}^{2}\\
+76902226329600\,t_{{1}}t_{{9}}+
52344714240000\,{t_{{1}}}^{3}t_{{7}}+10730666419200\,{t_{{1}}}^{5}t_{{
5}}\\
+992397296\,{t_{{1}}}^{10}{N}^{4}-124813568\,{N}^{6}{t_{{1}}}^{10}+
1053904737600\,{t_{{1}}}^{7}t_{{3}}+8439552\,{N}^{8}{t_{{1}}}^{10}\\
-3984998904
\,{t_{{1}}}^{10}{N}^{2}+26499511584000\,{t_{{1}}}^{4}{t_{{3}}}^{2}+
51293581516800\,t_{{1}}{t_{{3}}}^{3}\\
\left.-2571396710400\,{t_{{5}}}^{2}{N}^{
2}+29727129600\,{t_{{5}}}^{2}{N}^{4}+6220965465\,{t_{{1}}}^{10}-292864\,{N}^{10}{t_{{1}}}^{10}
\right)
\ee

%\appendix
\section{Free energy of generalized BGW model}
\label{C}
%\addcontentsline{toc}{section}{Appendix C}
\def\theequation{C\arabic{equation}}
\setcounter{equation}{0}

\be\label{F01}
\tilde{\mathcal F}_0^{(1)}=0
\ee

\be
\tilde{\mathcal F}_0^{(2)}=\frac{1}{2^3} T_{{3}}
\ee

\be
\tilde{\mathcal F}_0^{(3)}= \frac{1}{2^{5}\cdot 2}\left(2^2\cdot 3\,{T_{{3}}}^{2}+2^2\,T_{{5}}\right)
\ee

\be
\tilde{\mathcal F}_0^{(4)}=\frac{1}{2^{7}\cdot 3}\left(  3^{3} \cdot 5 \,T_{{5}}T_{{3}}+2^3\cdot 3^3\,{T_{{3}}}^{3}+3\cdot 5\,T_{{7}}  \right)
\ee

\be
\tilde{\mathcal F}_0^{(5)}=
\frac{1}{2^{9}\cdot 4}\left(2^5\cdot3^3\cdot 5\,{T_{{3}}}^{2}T_{{5}}+2^5\cdot 3\cdot7\,T_{{7}}T_{{3}}+2^4\cdot3^3\cdot 11\,{T_{{3}}}^{4}+2^3\cdot3^2\cdot 5
\,{T_{{5}}}^{2}+2^3\cdot7\,T_{{9}}
\right)
\ee

\be
\tilde{\mathcal F}_0^{(6)}=\frac{1}{2^{11}\cdot 5}\left(2\cdot3^2\cdot5^2\cdot7\,T_{{9}}T_{{3}}+2^2\cdot5^3\cdot 7\,T_{{7}}T_{{5}}+2^4\cdot3^2\cdot5^2\cdot 7\,{T_{{3}}}^{2}T_{{7}}\right.\\
\left.+
2^4\cdot3^3\cdot5^2\cdot13\,T_{{5}}{T_{{3}}}^{3}+2^3\cdot3^3\cdot5^3\,{T_{{5}}}^{2}T_{{3}}+2^4\cdot3^4\cdot 7\cdot13\,{T_{{
3}}}^{5}+2\cdot3\cdot5\cdot7\,T_{{11}}
\right)\\
\ee

\be
\tilde{\mathcal F}_0^{(7)}=\frac{1}{2^{13}\cdot 6}\left(2^6\cdot3^3\cdot5^2\cdot 7\,T_{{5}}T_{{3}}T_{{7}}+2^4\cdot3^6\cdot5^3\,{T_{{3}}}^{2}{T_{{5}}}^{2}+
2^4\cdot3^5\cdot5\cdot7\,{T_{{3}}}^{2}T_{{9}}\right.\\
+2^3\cdot3^4\cdot5^2\,T_{{9}}T_{{5}}+2^6\cdot3^4\cdot5^2\cdot7\,{T_{{3}}}^{
3}T_{{7}}+2^8\cdot3^6\cdot5^2\,T_{{5}}{T_{{3}}}^{4}+2^4\cdot3^4\cdot11\,T_{{11}}T_{{3}}\\
\left.+2^4\cdot3^3\cdot5^3\,
{T_{{5}}}^{3}+2^4\cdot3\cdot5^2\cdot7\,{T_{{7}}}^{2}+2^8\cdot3^6\cdot17\,{T_{{3}}}^{6}+2^3\cdot3^2\cdot11\,T_{{13}}
\right)\\
\ee

\be\label{F09}
\tilde{\mathcal F}_0^{(8)}=\frac{1}{2^{15}\cdot 7}\left(2^7\cdot3^7\cdot17\cdot19\,{T_{{3}}}^{
7}+2^4\cdot3^4\cdot5^2\cdot7^2\,T_{{5}}T_{{3}}T_{{9}}+2^5\cdot3^3\cdot5^2\cdot7^2\cdot17\,{T_{{3}}}^{2}T_{{7}}T_{{5}}\right.\\
+3^2\cdot7^2\cdot11\cdot13\,T_{{13}}T_{{3}}+2^5\cdot3^5\cdot5\cdot7^2\cdot17\,{T_{{3}}}^{4}T_{{7}}+2^4\cdot3^3\cdot5^3\cdot7\cdot17\,{T_{{5}
}}^{3}T_{{3}}\\
+3^3\cdot5\cdot7^2\cdot11\,T_{{11}}T_{{5}}+2^4\cdot3^4\cdot7^2\cdot11\,T_{{11}}{T_{{3}}}^{2}+
2^4\cdot3^2\cdot5^3\cdot7^2\,T_{{7}}{T_{{5}}}^{2}\\
+2^5\cdot3^6\cdot5^2\cdot7\cdot17\,{T_{{3}}}^{3}{T_{{5}}}^{2}+
2^5\cdot3\cdot5^2\cdot7^3\,{T_{{7}}}^{2}T_{{3}}+3^2\cdot5^2\cdot7^3\,T_{{9}}T_{{7}}\\
\left.+2^4\cdot3^4\cdot5\cdot7^2\cdot17\,{T_{{3}}}^
{3}T_{{9}}+2^5\cdot3^7\cdot7\cdot17\cdot19\,T_{{5}}{T_{{3}}}^{5}+
3\cdot7\cdot11\cdot13\,T_{{15}}
\right)
\ee

\be
\tilde{\mathcal F}_1^{(1)}={\frac {5}{16}}\,T_{{3}}
\ee

\be
\tilde{\mathcal F}_1^{(2)}={\frac {93}{64}}\,{T_{{3}}}^{2}+{\frac {35}{64}}\,T_{{5}}

\ee

\be
\tilde{\mathcal F}_1^{(3)}={\frac {75}{8}}\,{T_{{3}}}^{3}+{\frac {825}{128}}\,T_{{5}}T_{{3}}+{\frac {105}{128}}\,T_{{7}}
\ee

\be
\tilde{\mathcal F}_1^{(4)}={\frac {6225}{1024}}\,{T_{{5}}}^{2}+{\frac {35397}{512}}\,{T_{{3}}}^{4
}+{\frac {17415}{256}}\,{T_{{3}}}^{2}T_{{5}}+{\frac {3045}{256}}\,T_{{
7}}T_{{3}}+{\frac {1155}{1024}}\,T_{{9}}
\ee

\be
\tilde{\mathcal F}_1^{(5)}={\frac {140373}{256}}\,{T_{{3}}}^{5}+{\frac {40005}{2048}}\,T_{{9}}T_{
{3}}+{\frac {35595}{256}}\,{T_{{3}}}^{2}T_{{7}}+{\frac {20825}{1024}}
\,T_{{7}}T_{{5}}+{\frac {73125}{512}}\,{T_{{5}}}^{2}T_{{3}}\\
+{\frac {
178875}{256}}\,{T_{{3}}}^{3}T_{{5}}+{\frac {3003}{2048}}\,T_{{11}}
\ee

\be
\tilde{\mathcal F}_1^{(6)}=
{\frac {64925}{4096}}\,{T_{{7}}}^{2}+{\frac {1165509}{256}}\,{T_{{3}}}
^{6}+{\frac {373875}{4096}}\,{T_{{5}}}^{3}+{\frac {
1036665}{4096}}\,T_{{9}}{T_{{3}}}^{2}+{\frac {542325}{1024}}\,T_{{5}}T
_{{3}}T_{{7}}\\+{\frac {15015}{8192}}\,T_{{13}}
+{\frac {1552635}{1024}}\,T_
{{7}}{T_{{3}}}^{3}+{\frac {256725}{8192}}\,T_{{9}}T_{{5}}+{\frac {
121275}{4096}}\,T_{{11}}T_{{3}}+{\frac {9605925}{4096}}\,{T_{{5}}}^{2}
{T_{{3}}}^{2}\\
+{\frac {1819665}{256}}\,{T_{{3}}}^{4}T_{{5}}
\ee

\be
\tilde{\mathcal F}_1^{(7)}={\frac {69903081}{1792}}\,{T_{{3}}}^{7}+{\frac {513135}{32}}\,T_{{7}}{
T_{{3}}}^{4}+{\frac {921375}{1024}}\,T
_{{5}}T_{{3}}T_{{9}}+{\frac {9520875}{1024}}\,T_{{5}}{T_{{3}}}^{2}T_{{
7}}\\
+{\frac {466725}{1024}}\,{T_{{7}}}^{2}T_{{3}}+{\frac {
6041385}{2048}}\,T_{{9}}{T_{{3}}}^{3}+{\frac {693693}{16384}}\,T_{{13}
}T_{{3}}+{\frac {744975}{16384}}\,T_{{11}}T_{{5}}+{\frac {866943}{2048
}}\,{T_{{3}}}^{2}T_{{11}}\\
+{\frac {33969375}{1024}}\,{T_{{5}}}^{2}{T_{{
3}}}^{3}+{\frac {9218205}{128}}\,{T_{{3}}}^{5}T_{{5}}+{\frac {3290625}
{1024}}\,{T_{{5}}}^{3}T_{{3}}+{\frac {760725}{16384}}\,T_{{9}}T_{{7}}\\
+{\frac {485625}{1024}}\,{T_{{5}}}^{2}T_{{7}}+{\frac {36465}{16384}}\,T_{{15}}
\ee

\be
\tilde{\mathcal F}_2^{(1)}={\frac {259}{256}}\,T_{{5}}+{\frac {657}{256}}\,{T_{{3}}}^{2}
\ee

\be
\tilde{\mathcal F}_2^{(2)}={\frac {6201}{128}}\,{T_{{3}}}^{3}+{\frac {36015}{1024}}\,T_{{5}}T_{{3
}}+{\frac {4935}{1024}}\,T_{{7}}
\ee

\be
\tilde{\mathcal F}_2^{(3)}={\frac {74529}{512}}\,T_{{7}}T_{{3}}+{\frac {397035}{512}}\,{T_{{3}}}^
{2}T_{{5}}+{\frac {765693}{1024}}\,{T_{{3}}}^{4}+{\frac {30723}{2048}}
\,T_{{9}}+{\frac {149985}{2048}}\,{T_{{5}}}^{2}
\ee

\be
\tilde{\mathcal F}_2^{(4)}={\frac {7390845}{16384}}\,T_{{9}}T_{{3}}+{\frac {28598265}{2048}}\,{T_
{{3}}}^{3}T_{{5}}+{\frac {6086745}{2048}}\,{T_{{3}}}^{2}T_{{7}}+{
\frac {12284325}{4096}}\,{T_{{5}}}^{2}T_{{3}}\\
+{\frac {3744825}{8192}}
\,T_{{7}}T_{{5}}
+{\frac {106851717}{10240}}\,{T_{{3}}}^{5}+{\frac {
603603}{16384}}\,T_{{11}}
\ee

\be
\tilde{\mathcal F}_2^{(5)}={\frac {9707635}{16384}}\,{T_{{7}}}^{2}+{\frac {51448425}{16384}}\,{T_
{{5}}}^{3}+{\frac {140271507}{1024}}\,{T_{{3}}}^{6}\\
+{\frac {229158315}
{1024}}\,{T_{{3}}}^{4}T_{{5}}+{\frac {1265323275}{16384}}\,{T_{{5}}}^{
2}{T_{{3}}}^{2}+{\frac {38701215}{32768}}\,T_{{9}}T_{{5}}+{\frac {
149588775}{16384}}\,T_{{9}}{T_{{3}}}^{2}\\
+{\frac {208346985}{4096}}\,T_
{{7}}{T_{{3}}}^{3}+{\frac {18927909}{16384}}\,T_{{11}}T_{{3}}+{\frac {
76102635}{4096}}\,T_{{5}}T_{{3}}T_{{7}}+{\frac {2543541}{32768}}\,T_{{
13}}
\ee

\be
\tilde{\mathcal F}_2^{(6)}={\frac {441784935}{256}}\,{T_{{3}}}^{7}+{\frac {3217790205}{4096}}\,T_
{{7}}{T_{{3}}}^{4}+{\frac {50498175}{2048}}\,{T_{{7}}}^{2}T_{{3}}+{
\frac {634027905}{4096}}\,T_{{9}}{T_{{3}}}^{3}\\
+{\frac {338585247}{
131072}}\,T_{{13}}T_{{3}}+{\frac {349538805}{131072}}\,T_{{11}}T_{{5}}
+{\frac {97956243}{4096}}\,{T_{{3}}}^{2}T_{{11}}+{\frac {3266594325}{
2048}}\,{T_{{5}}}^{2}{T_{{3}}}^{3}\\
+{\frac {13610282013}{4096}}\,{T_{{3
}}}^{5}T_{{5}}+{\frac {1317331125}{8192}}\,{T_{{5}}}^{3}T_{{3}}+{
\frac {351871695}{131072}}\,T_{{9}}T_{{7}}+{\frac {205991625}{8192}}\,
{T_{{5}}}^{2}T_{{7}}\\
+{\frac {402291225}{8192}}\,T_{{5}}T_{{3}}T_{{9}}+
{\frac {971414325}{2048}}\,T_{{5}}{T_{{3}}}^{2}T_{{7}}+{\frac {
19246227}{131072}}\,T_{{15}}
\ee

\be
\tilde{\mathcal F}_3^{(1)}=75\,{T_{{3}}}^{3}+{\frac {114225}{2048}}\,T_{{5}}T_{{3}}+{\frac {16145
}{2048}}\,T_{{7}}
\ee

\be
\tilde{\mathcal F}_3^{(2)}={\frac {1399965}{2048}}\,T_{{7}}T_{{3}}+{\frac {7170255}{2048}}\,{T_{{
3}}}^{2}T_{{5}}+{\frac {13407093}{4096}}\,{T_{{3}}}^{4}+{\frac {604835
}{8192}}\,T_{{9}}+{\frac {2804625}{8192}}\,{T_{{5}}}^{2}
\ee

\be
\tilde{\mathcal F}_3^{(3)}={\frac {75131595}{16384}}\,T_{{9}}T_{{3}}+{\frac {265649625}{2048}}\,{
T_{{3}}}^{3}T_{{5}}+{\frac {58951305}{2048}}\,{T_{{3}}}^{2}T_{{7}}+{
\frac {118222875}{4096}}\,{T_{{5}}}^{2}T_{{3}}\\
+{\frac {37711975}{8192}
}\,T_{{7}}T_{{5}}+{\frac {192117987}{2048}}\,{T_{{3}}}^{5}+{\frac {
6483477}{16384}}\,T_{{11}}
\ee

\be
\tilde{\mathcal F}_3^{(4)}={\frac {711506425}{65536}}\,{T_{{7}}}^{2}+{\frac {5317009425}{16384}}\,T_{{5}}T_{{3}}T_{{7}}+{
\frac {201052995}{131072}}\,T_{{13}}\\
+{\frac {3562329375}{65536}}
\,{T_{{5}}}^{3}+{\frac {8767078173}{4096}}\,{T_{{3}}}^{6}+{\frac {
14797367145}{4096}}\,{T_{{3}}}^{4}T_{{5}}+{\frac {84554265825}{65536}}
\,{T_{{5}}}^{2}{T_{{3}}}^{2}\\
+{\frac {2843426025}{131072}}\,T_{{9}}T_{{
5}}+{\frac {10571621445}{65536}}\,T_{{9}}{T_{{3}}}^{2}+{\frac {
14033060055}{16384}}\,T_{{7}}{T_{{3}}}^{3}+{\frac {1411098975}{65536}}
\,T_{{11}}T_{{3}}
\ee

\be
\tilde{\mathcal F}_3^{(5)}=
{\frac {1211928290217}{28672}}\,{T_{{3}}}^{7}+{\frac {84821501835}{
4096}}\,T_{{7}}{T_{{3}}}^{4}
+{\frac {11533969125}{16384}}\,{T_{{7}}}^{
2}T_{{3}}\\
+{\frac {140161531545}{32768}}\,T_{{9}}{T_{{3}}}^{3}+{\frac {
20835465651}{262144}}\,T_{{13}}T_{{3}}+{\frac {21099663585}{262144}}\,
T_{{11}}T_{{5}}\\
+{\frac {22803571791}{32768}}\,{T_{{3}}}^{2}T_{{11}}+{
\frac {682677669375}{16384}}\,{T_{{5}}}^{2}{T_{{3}}}^{3}+{\frac {
344109377925}{4096}}\,{T_{{3}}}^{5}T_{{5}}\\
+{\frac {71155771875}{16384}
}\,{T_{{5}}}^{3}T_{{3}}+{\frac {21140849835}{262144}}\,T_{{9}}T_{{7}}+
{\frac {11630143875}{16384}}\,{T_{{5}}}^{2}T_{{7}}\\
+{\frac {23039048025
}{16384}}\,T_{{5}}T_{{3}}T_{{9}}+{\frac {211982823675}{16384}}\,T_{{5}
}{T_{{3}}}^{2}T_{{7}}+{\frac {1256693295}{262144}}\,T_{{15}}
\ee

\be
\tilde{\mathcal F}_4^{(1)}={\frac {16776921}{16384}}\,T_{{7}}T_{{3}}+{\frac {84428595}{16384}}\,{
T_{{3}}}^{2}T_{{5}}
+{\frac {155619117}{32768}}\,{T_{{3}}}^{4}+{\frac {
7400547}{65536}}\,T_{{9}}+{\frac {33567585}{65536}}\,{T_{{5}}}^{2}
\ee

\be
\tilde{\mathcal F}_4^{(2)}={\frac {5177752965}{262144}}\,T_{{9}}T_{{3}}+{\frac {61659550053}{163840}}\,{T
_{{3}}}^{5}+{\frac {461311851}{262144}}\,T_{{11}}\\
+{\frac {17389528605}{
32768}}\,{T_{{3}}}^{3}T_{{5}}+{\frac {3952037565}{32768}}\,{T_{{3}}}^{
2}T_{{7}}+{\frac {7910056125}{65536}}\,{T_{{5}}}^{2}T_{{3}}+{\frac {
2591640625}{131072}}\,T_{{7}}T_{{5}}
\ee

\be
\tilde{\mathcal F}_4^{(3)}={\frac {25761027005}{262144}}\,{T_{{7}}}^{2}+{\frac {124408920975}{
262144}}\,{T_{{5}}}^{3}\\
+{\frac {284951872593}{16384}}\,{T_{{3}}}^{6}+{
\frac {491651138745}{16384}}\,{T_{{3}}}^{4}T_{{5}}+{\frac {
2877392953125}{262144}}\,{T_{{5}}}^{2}{T_{{3}}}^{2}\\
+{\frac {
103019276745}{524288}}\,T_{{9}}T_{{5}}+{\frac {371988099945}{262144}}
\,T_{{9}}{T_{{3}}}^{2}+{\frac {478941497055}{65536}}\,T_{{7}}{T_{{3}}}
^{3}\\
+{\frac {51393633291}{262144}}\,T_{{11}}T_{{3}}+{\frac {
186325900005}{65536}}\,T_{{5}}T_{{3}}T_{{7}}+{\frac {7612223619}{
524288}}\,T_{{13}}

\ee

\be
\tilde{\mathcal F}_4^{(4)}=
{\frac {19375419429891}{32768}}\,{T_{{3}}}^{7}+{\frac {39967361286615}
{131072}}\,T_{{7}}{T_{{3}}}^{4}+{\frac {1441907245275}{131072}}\,{T_{{
7}}}^{2}T_{{3}}\\
+{\frac {17044272862155}{262144}}\,T_{{9}}{T_{{3}}}^{3}
+{\frac {5457509668611}{4194304}}\,T_{{13}}T_{{3}}+{\frac {
5481470419425}{4194304}}\,T_{{11}}T_{{5}}\\
+{\frac {2872551695013}{
262144}}\,{T_{{3}}}^{2}T_{{11}}+{\frac {80106396840225}{131072}}\,{T_{
{5}}}^{2}{T_{{3}}}^{3}+{\frac {157614682018959}{131072}}\,{T_{{3}}}^{5
}T_{{5}}\\
+{\frac {17136214183125}{262144}}\,{T_{{5}}}^{3}T_{{3}}+{
\frac {5484436521795}{4194304}}\,T_{{9}}T_{{7}}+{\frac {2892529880625}
{262144}}\,{T_{{5}}}^{2}T_{{7}}\\
+{\frac {5765593974525}{262144}}\,T_{{5
}}T_{{3}}T_{{9}}+{\frac {25639814092725}{131072}}\,T_{{5}}{T_{{3}}}^{2
}T_{{7}}+{\frac {343519458711}{4194304}}\,T_{{15}}
\ee

\be
\tilde{\mathcal F}_5^{(1)}={\frac {14960246805}{524288}}\,T_{{9}}T_{{3}}+{\frac {49062715875}{
65536}}\,{T_{{3}}}^{3}T_{{5}}+{\frac {11274351795}{65536}}\,{T_{{3}}}^
{2}T_{{7}}\\
+{\frac {22552975125}{131072}}\,{T_{{5}}}^{2}T_{{3}}+{\frac 
{7482105225}{262144}}\,T_{{7}}T_{{5}}+{\frac {34468789653}{65536}}\,{T
_{{3}}}^{5}+{\frac {1352576043}{524288}}\,T_{{11}}
\ee

\be
\tilde{\mathcal F}_5^{(2)}={\frac {418486752775}{1048576}}\,{T_{{7}}}^{2}+{\frac {1979457545625}{
1048576}}\,{T_{{5}}}^{3}+{\frac {4336231722327}{65536}}\,{T_{{3}}}^{6}\\
+{\frac {7582592016315}{65536}}\,{T_{{3}}}^{4}T_{{5}}+{\frac {
45038717337375}{1048576}}\,{T_{{5}}}^{2}{T_{{3}}}^{2}+{\frac {
1673842629375}{2097152}}\,T_{{9}}T_{{5}}\\
+{\frac {5932794319515}{
1048576}}\,T_{{9}}{T_{{3}}}^{2}+{\frac {7503673495185}{262144}}\,T_{{7
}}{T_{{3}}}^{3}+{\frac {836393983425}{1048576}}\,T_{{11}}T_{{3}}\\
+{
\frac {2967884007375}{262144}}\,T_{{5}}T_{{3}}T_{{7}}+{\frac {
126762200565}{2097152}}\,T_{{13}}
\ee

\be
\tilde{\mathcal F}_5^{(3)}={\frac {295477654037589}{65536}}\,{T_{{3}}}^{7}+{\frac {78832872717795
}{32768}}\,T_{{7}}{T_{{3}}}^{4}+{\frac {23702887069575}{262144}}\,{T_{
{7}}}^{2}T_{{3}}\\
+{\frac {274556802988755}{524288}}\,T_{{9}}{T_{{3}}}^{
3}+{\frac {46148268185019}{4194304}}\,T_{{13}}T_{{3}}+{\frac {
46210569510105}{4194304}}\,T_{{11}}T_{{5}}\\
+{\frac {47349163650789}{
524288}}\,{T_{{3}}}^{2}T_{{11}}+{\frac {1262180334313125}{262144}}\,{T
_{{5}}}^{2}{T_{{3}}}^{3}+{\frac {152599213285155}{16384}}\,{T_{{3}}}^{
5}T_{{5}}\\
+{\frac {137510048829375}{262144}}\,{T_{{5}}}^{3}T_{{3}}+{
\frac {46217234012355}{4194304}}\,T_{{9}}T_{{7}}+{\frac {
23725543608375}{262144}}\,{T_{{5}}}^{2}T_{{7}}\\
+{\frac {47401263300825}
{262144}}\,T_{{5}}T_{{3}}T_{{9}}+{\frac {412202168102025}{262144}}\,T_
{{5}}{T_{{3}}}^{2}T_{{7}}+{\frac {2989207836615}{4194304}}\,T_{{15}}
\ee

\be
\tilde{\mathcal F}_6^{(1)}={\frac {2356605625185}{4194304}}\,{T_{{7}}}^{2}+{\frac {11038896821475
}{4194304}}\,{T_{{5}}}^{3}+{\frac {23660311883769}{262144}}\,{T_{{3}}}
^{6}\\
+{\frac {41645338352865}{262144}}\,{T_{{3}}}^{4}T_{{5}}+{\frac {
249159522789225}{4194304}}\,{T_{{5}}}^{2}{T_{{3}}}^{2}+{\frac {
9426278461365}{8388608}}\,T_{{9}}T_{{5}}\\
+{\frac {33108811628445}{
4194304}}\,T_{{9}}{T_{{3}}}^{2}+{\frac {41522675471235}{1048576}}\,T_{
{7}}{T_{{3}}}^{3}+{\frac {4712392198503}{4194304}}\,T_{{11}}T_{{3}}\\
+{
\frac {16556502355785}{1048576}}\,T_{{5}}T_{{3}}T_{{7}}+{\frac {
721976952807}{8388608}}\,T_{{13}}

\ee

\be
\tilde{\mathcal F}_6^{(2)}={\frac {15185336065870311}{917504}}\,{T_{{3}}}^{7}+{\frac {
9454513725058185}{1048576}}\,T_{{7}}{T_{{3}}}^{4}+{\frac {
11387713000125}{32768}}\,{T_{{7}}}^{2}T_{{3}}\\
+{\frac {1041752546225055
}{524288}}\,T_{{9}}{T_{{3}}}^{3}+{\frac {1442169677387439}{33554432}}
\,T_{{13}}T_{{3}}+{\frac {1442721973704165}{33554432}}\,T_{{11}}T_{{5}
}\\
+{\frac {182141818558623}{524288}}\,{T_{{3}}}^{2}T_{{11}}+{\frac {
1182045374663625}{65536}}\,{T_{{5}}}^{2}{T_{{3}}}^{3}+{\frac {
36199921141511001}{1048576}}\,{T_{{3}}}^{5}T_{{5}}\\
+{\frac {
4169029399318125}{2097152}}\,{T_{{5}}}^{3}T_{{3}}+{\frac {
1442776147146495}{33554432}}\,T_{{9}}T_{{7}}+{\frac {729014470760625}{
2097152}}\,{T_{{5}}}^{2}T_{{7}}\\
+{\frac {1457590887545625}{2097152}}\,T
_{{5}}T_{{3}}T_{{9}}+{\frac {781513453641375}{131072}}\,T_{{5}}{T_{{3}
}}^{2}T_{{7}}+{\frac {95041284259779}{33554432}}\,T_{{15}}
\ee

\be
\tilde{\mathcal F}_7^{(1)}={\frac {23275794518914797}{1048576}}\,{T_{{3}}}^{7}+{\frac {
12810322271682495}{1048576}}\,T_{{7}}{T_{{3}}}^{4}+{\frac {
1999590270994575}{4194304}}\,{T_{{7}}}^{2}T_{{3}}\\
+{\frac {
22719742197138315}{8388608}}\,T_{{9}}{T_{{3}}}^{3}+{\frac {
3987196321745637}{67108864}}\,T_{{13}}T_{{3}}+{\frac {3987584386351335
}{67108864}}\,T_{{11}}T_{{5}}\\
+{\frac {3998836418130477}{8388608}}\,{T_
{{3}}}^{2}T_{{11}}+{\frac {102487735548028125}{4194304}}\,{T_{{5}}}^{2
}{T_{{3}}}^{3}+{\frac {48784570534530045}{1048576}}\,{T_{{3}}}^{5}T_{{
5}}\\
+{\frac {11361282083218125}{4194304}}\,{T_{{5}}}^{3}T_{{3}}+{\frac 
{3987621333347085}{67108864}}\,T_{{9}}T_{{7}}+{\frac {1999731386512125
}{4194304}}\,{T_{{5}}}^{2}T_{{7}}\\
+{\frac {3999155802479775}{4194304}}
\,T_{{5}}T_{{3}}T_{{9}}+{\frac {34081835456653425}{4194304}}\,T_{{5}}{
T_{{3}}}^{2}T_{{7}}+{\frac {264952094603625}{67108864}}\,T_{{15}}
\ee

%\appendix
\section{Free energy of generalized BGW model as a linear combination of $B_k(N)$}
\label{D}
%\addcontentsline{toc}{section}{Appendix D}
\def\theequation{D\arabic{equation}}
\setcounter{equation}{0}

\be
{\mathcal F}_{N}^{(2)}=\frac{1}{2^7} B_2(N) T_3
\ee

\be
{\mathcal F}_{N}^{(3)}={\frac {1}{2^{10}}}\, B_3(N) \left( T_5+3 T_3^2\right) -\frac{3}{2^{8}}B_2(N) T_3^2
\ee

\be
{\mathcal F}_{N}^{(4)}={\frac {1}{2^{15}}}\,B_4(N)  \left(5 T_7 +3^2\cdot5  T_3T_5 +2^3\cdot3^2 T_3^3\right)
-\frac{3}{2^{10}}B_3(N)\left( 5T_3T_5+13 T_3^3\right)
+ \frac{3}{2^7}B_2(N) T_3^3
\ee

\be
{\mathcal F}_{N}^{(5)}={\frac {1}{2^{18}}}\,  B_5(N) \left(7\,T_{{9}}+3^2\cdot 5 \,{T_{{5}}}^{2}+2^2\cdot 3^3\cdot 5 \,{T_{{3}}}^{2}
T_{{5}}+2^2\cdot 3\cdot 7 \,T_{{7}}T_{{3}}+2\cdot 3^3\cdot 11\,{T_{{3}}}^{4}
\right)\\
-\frac{3}{2^{14}}B_4(N)\left(5^2\,{T_{{5}}}^{2}+5\cdot 7\,T_{{7}}T_{{3}}+3^4\cdot 5\,{T_{{3}}}^{2}T_{{5}}+3^4\cdot 7\,{T
_{{3}}}^{4}
\right)\\
+\frac{3^2}{2^{11}}B_3(N)\left(5T_5^2+3^2\cdot 13 T_3^4+2^2\cdot 3\cdot5 T_3^2 T_5\right)
-\frac{3^3}{2^9}B_2(N)T_3^4
\ee

\be
{\mathcal F}_{N}^{(6)}={\frac {1}{2^{22}\cdot 5}}\,B_6(N) \left(3\cdot 5\cdot 7\, T_{11}+
3^2\cdot5^2\cdot7\,T_{{9}}T_{{3}}+2^3\cdot3^4\cdot7\cdot13\,{T_{{3}}}^{5}\right.\\
\left.+2^3\cdot3^2\cdot5^2\cdot7\,{T_{{3}}}^{2}T_{{7}}+
2^2\cdot3^3\cdot5^3\,{T_{{5}}}^{2}T_{{3}}+2\cdot5^3\cdot7\,T_{{7}}T_{{5}}+2^3\cdot3^3\cdot5^2\cdot13\,
T_{{5}}{T_{{3}}}^{3}
\right)\\
-\frac{3}{2^{17}\cdot 5}B_5(N)\left(5^3\cdot 7\,T_{{7}}T_{{5}}+3^2\cdot5^3\cdot 7\,{T_{{3}}}^{2}T_{{7}}+3\cdot5^2\cdot7\,T_{{9}}T_{{3}}+
3^3\cdot5^2\cdot 89\,T_{{5}}{T_{{3}}}^{3}\right.\\
\left.+2\cdot3\cdot5^3\cdot13\,{T_{{5}}}^{2}T_{{3}}+2\cdot3^3\cdot1087\,{T_{{3}}
}^{5}
\right)
+\frac{3^2}{2^{15}\cdot 5}B_4(N)\left(5^3\cdot 7\,T_{{7}}T_{{5}}+2\cdot3\cdot5^3\cdot7\,{T_{{3}}}^{2}T_{{7}}\right.\\
\left.+2^2\cdot3^3\cdot5^4\,T_{{5}}{T_{{3}}}
^{3}+2^2\cdot3\cdot5^3\cdot7\,{T_{{5}}}^{2}T_{{3}}+3^4\cdot31^2\,{T_{{3}}}^{5}
\right)
-\frac{3^3}{2^{10}\cdot 5}B_3(N) \left(5^3\cdot 7 T_3^2 T_5\right.\\
\left.+5^3 T_5^2+2\cdot 3\cdot 229 T_3^5\right)T_3
+\frac{3^4}{2^7\cdot 5}\, B_2(N)T_3^5
\ee

%\be
%{\mathcal F}_{N}^{(7)}=
%\ee

\end{appendices}

\end{document}